\RequirePackage{fix-cm}
\documentclass[smallcondensed]{svjour3}     % onecolumn (ditto)

\smartqed  % flush right qed marks, e.g. at end of proof
\usepackage{graphicx,url,amsmath,hyperref}
\usepackage{natbib}
\begin{document}

\title{Predictive distributions that mimic frequencies \\
over a restricted subdomain
%\thanks{Expanded Version}
%\thanks{Grants or other notes
%about the article that should go on the front page should be
%placed here. General acknowledgments should be placed at the end of the article.}
}
\subtitle{Expanded preprint version}

%\titlerunning{Predictive distributions that mimic frequencies over a restricted subdomain}        % if too long for running head

\author{Frank Lad        \and
        Giuseppe Sanfilippo %etc.
}

%\authorrunning{Short form of author list} % if too long for running head

\institute{Frank Lad \at
              Department of Mathematics and Statistics  \\
              University of Canterbury, Christchurch, New Zealand \\
              \email{F.Lad@math.canterbury.ac.nz}           %  \\
%             \emph{Present address:} of F. Author  %  if needed
           \and
           Giuseppe Sanfilippo \at
             Department of Mathematics and Computer Science\\
             Via Archirafi 34, 90123 Palermo, Italy\\
             University of Palermo, Italy
             \email{giuseppe.sanfilippo@unipa.it}
}

\date{Received: date / Accepted: date}
% The correct dates will be entered by the editor

\maketitle

\begin{abstract}
\noindent A predictive distribution over a sequence of $N+1$ events
is said to be ``frequency mimicking" whenever the
probability for the final event conditioned on the outcome of the
first $N$ events equals the relative frequency of successes among
them. Infinitely extendible exchangeable distributions that {\it
universally} inhere this property 
are known to have several annoying concomitant properties. We
motivate frequency mimicking assertions {\it
over a limited subdomain} in practical problems of
finite inference, and we identify their computable coherent implications. 
We provide some computed examples using
reference distributions, and we introduce computational software to generate
any specification. The software derives from an inversion of the
finite form of the exchangeability representation theorem.  Three new theorems 
delineate the extent of the usefulness of such
distributions, and we show why it may {\it not} be appropriate to
extend the frequency mimicking assertions for a specified value of $N$ 
to any arbitrary larger size of $N$. The constructive results identify 
the source and structure of ``adherent masses'' in the limit of a sequence of
finitely additive distributions.  Appendices develop a novel
geometrical representation of conditional probabilities which
illuminate the analysis.

%Insert your abstract here. Include keywords, PACS and mathematical
%subject classification numbers as needed.
%\keywords{First keyword \and Second keyword \and More}

% \PACS{PACS code1 \and PACS code2 \and more}
% \subclass{MSC code1 \and MSC code2 \and more}
\end{abstract}

\section{\ Introduction}
\label{sect:intro}
The subjectivist understanding of probability accepts naturally the
notion that conditional probabilities can be asserted as prior
information for the analysis of problems involving uncertainty.
Bounds on associated unconditional probabilities can be derived from
these, using the principle of coherence governing all assertions.
The theoretical basis for this approach lies in de Finetti's
construction of conditional probability as a price for a contingent
transaction, and his fundamental theorem of prevision (FTP). See \citet*{Lad90,Lad92} and \cite[2.10, 3.3]{Lad96}.
%Lad, Dickey and Rahman (1990, 1992) and Lad (1996a, 2.10, 3.3).
Applications of interest can be found in the articles of 
\citet*[Section VI]{Johnson05} and of \citet*{Capotorti07}.
%Johnson,
%Moosman and Cotter (2005, Section VI) and of Capotorti, Lad and Sanfilippo (2007). 
The former uses expectations regarding successful
commercial rocket launches under a variety of conditions to assert
``prior'' knowledge via conditional probabilities.  The latter
orders conditional probabilities based on natural attitudes toward
``median medical diagnoses'' derived from several examining radiologists who are blinded to the assessments of one another.  These
orderings are used along with other forms of partial knowledge to
compute bounds on accuracy rates of median diagnoses, applying
an extension of the FTP to quadratic conditions.  It has also been extended to include situations where the conditioning event may be assessed with probability zero.  See \citet*{Biazzo00,Capotorti03,Coletti96,Gilio16,Regazzini87}.
% Biazzo and Gilio (2000), Gilio et {\it al.} (2016),
%Regazzini (1987), and Coletti and Scozzafava (1996).
%\todo[inline]{I would add something like ``The FTP has been suitable extended  to manage 
%the case where conditioning event may have zero probabilities (see, e.g., Biazzo and Gilio (2000), Gilio et {\it al.} (2016),
%Regazzini(1987), Coletti and Scozzafava (1996).  ''\\
%References\\
%\noindent{\bf Biazzo V. and Gilio, A.}(2000) A generalization of the fundamental theorem of de Finetti for imprecise conditional probability assessments. {\it Internat. J. Approx. Reason.} 251-272. \\
%\noindent {\bf Coletti, G. and Scozzafava, R.} (1996) Characterization of coherent conditional probabilities as a tool for their assessment and extension. {\it International Journal of Uncertainty, Fuzziness and Knowledge-Based Systems}, {\bf 04}  103--127.\\
%\noindent {\bf Gilio, A., Pfeifer, N., and Sanfilippo, G.} (2016). Transitivity in coherence-based probability logic. {\it Journal of Applied Logic} {\bf 14} 46--64.\\
%\noindent {\bf Regazzini, E.} (1987) de Finetti's coherence and statistical inference. {\it Annals of Statistics} {\bf 15 (2)} 845--864.\\
%}

The present article investigates the use of conditional
probabilities as assertions of prior information in a context that
has a broad range of applicability, with particular reference to
probability elicitation. Objectivist statistical methods have long
promoted the sample frequency of occurrences as an appropriate
``estimate of the probability of an event'' in a
string of observations construed as independent Bernoulli. In
contrast, Bayesians typically call for adjustments to the sample
frequency as their ``posterior predictive probabilities for the next
event'', based on prior information in sampling setups they regard
exchangeably. Foundational differences aside however, there is
coherent support for more agreement in standard practice than the wide variation permitted by formal theoretical
comparisons. While pleased with their sample mean estimate ``when it
sounds reasonable'', frequentists tend to ``doubt the data'' when a
seemingly unusual sample result occurs, often wishing to repeat the
sampling experiment while informally hedging their bets.
Alternately, when a Bayesian hears of a sample statistic regarding a
matter one has not been thinking about, it is not uncommon in
practice to locate an expectation for the next observation at the
announced sample frequency.  Only when an announced frequency sounds
seriously out of bounds might one bother to adjust it when asserting
a predictive probability. It is the formalisation and assessment 
of this shared coherent practice that we address in this article.

Our results can be appreciated in the tradition of Savage, who once
wrote an unpublished exposition entitled ``The subjective basis of statistical
practice'' \citep{Savage61}.%(Savage, 1961). 
 Despite his pathbreaking investigations
and his collaboration with de Finetti, his deferential attitude to
aspects of the statistical practice established at his time has been
notable 
(\citealp[p4]{Savage54};
 \citealp[preface]{Savage72}). %(Savage, 1954, p4; 1972, preface). 
 The results we present
here give some support to his conviction that informal applied
activities of common practice do have a coherent foundation. 

Discussion of a motivating example can be simplified if we first
introduce a notation for probabilities, conditional probabilities,
and vectors of them.  We follow 
de Finetti's (\citeyear{deFinetti67,deFi70}) %(\citeyear{deFinetti67,deFinetti72})
%(1967, 1972)
 convention
of defining events as numbers rather than sets. Standard
set-theoretic formalists may read our allusions to ``events''  as
``indicators of events'' (random variables) without distortion.

\noindent {\bf Notation:}     Consider a sequence of $N+1$
events $E_1,E_2,\ldots,E_{N+1}$ that are regarded exchangeably.  For each $K=1,\ldots,N+1$, let $S_K$ denote the sum of the first $K$ of them, and $\bar{S}_K$ their average: $S_K=\sum_{i=1}^KE_{i}$,  and $\bar{S}_K=\frac{S_K}{K}$. Let
$q_{a,N+1}$ denote a probability for the sum, $P(S_{N+1} = a)$ for
$a = 0, 1, ..., N+1$; \ and let $p_{a,N}$ denote the conditional
probability $P(E_{N+1}|S_N = a)$ \; for $a = 0, 1, 2, ..., N$. Bold
letters denote vectors of these variables, as in the probability
mass function (pmf) vector {\bf q}$_{N+2} = (q_{0,N+1}, q_{1,N+1}, ...,
q_{N+1,N+1})$ and the conditional probability function vector {\bf
p}$_{N+1} = (p_{0,N}, p_{1,N}, ..., p_{N,N})$. The subscript on a
bold letter denotes the dimension of the
vector.  

In these terms, a conditional probability is said to mimic a
frequency whenever $p_{a,N} = a/N$. It was suggested in a
preliminary investigation by \citet{Lad96b} %Lad (1996b)
  that assertions of $p_{a,N}
= a/N$ {\it over a limited domain} of ``$a$'' values and a specific
size of $N$ may well represent forecasters' attitudes towards
sequences of
experimental observations.  

\noindent {\bf Example:}  Consider the events observed in a sample
of beehives from a large and productive apiary with some thousands
of hives. Define the event $E_i$ as the indicator that hive $i$ is
found to have at least one swarm cell (a new queen cell) in the brood box on an
inspection day, the inspection day being 1 full month into a lush
Spring.  The apiarist is
uncertain about whether any particular hive will be observed to have
new queen cells formed within its brood box, and regards such
events exchangeably over the entire apiary. Wishing to assess a probability distribution for the total number of hives containing at least one swarm cell, it is felt that a sample of
$100$ observations would contain enough information to locate a predictive 
probability for the next hive if the observed frequency in the sample were between
$.25$ and $.60$. Although this experienced beekeeper definitely has
informed opinions regarding bee-swarming behavior this Spring, there
is considered not to be enough prior information to motivate
adjusting the sample frequency as the predictive probability if the 
mean observation were found to be within this acceptable range. However, were
the sample frequency among 100 hives to be observed outside this interval, the
apiarist would want to adjust it toward the interval when asserting
a predictive probability. Furthermore, the predictive probability
may be specified never to fall below a lower bound such as $.10$ nor
above an upper bound such as $.70$ even if the observed sum of hives
with
queen cells were found to be as low as $0$ or as high as $100$. 

\noindent {\bf Assertion structure:}  In this context, the apiarist
may express prior knowledge in the
form of three types of assertions:
\\
\indent i.)  $p_{a,100} = a/100$ for integers ``$a$'' within the
interval
$[25, 60]$; \\
\indent  ii.)  $p_{a,100} \leq p_{a+1, 100}$ for integers ``$a$''
within the
intervals $[0, 24]$ and $[60, 99]$; \ \ \ and \\
\indent iii.) $p_{0,100} \geq .10$, and $p_{100, 100} \leq .70$. 

These assertions represent prior knowledge that places a sharply
defined interval, $[.25, .60]$, over reasonable values for
predictive probabilities based on observed frequencies of success
among preceding events in the sequence.  They do {\it not}
constitute {\it certainty} that the frequency {\it does lie} within
this interval. They only specify that prior information is not
refined enough to motivate adjusting a prospective frequency among
100 hives for inference about the 101$^{st}$ hive if the frequency
lies within this ``reasonable sounding'' interval. The predictive
probability would be adjusted away from the sample mean only if the
mean were to be observed outside this interval. Assertions in this
form are specific enough to imply computable bounds on associated 
probability mass functions for sums of any number of
events, such as {\bf q}$_{102}$, {\bf
q}$_{1002}$, or even {\bf q}$_{100002}$.  Similarly, they motivate 
bounds on conditional probabilities based on observed sequences of 
any other length. 

The formal assertions of the apiarist (i, ii and iii) amount to a
specification of a family of distributions we shall call {\it
frequency mimicking distributions}, or FMD's.  The assertions are not
sufficient to specify a unique complete distribution, since they
specify only $36$ linear equalities along with 65 inequalities among
the components of {\bf q}$_{102}$, which are restricted otherwise
only to sum to $1$. Thus, the family of cohering FMD's consists of a
$65$-dimensional polytope within the unit-simplex ${\bf S}^{101}$.
(This dimension derives from the dimension of {\bf q}$_{102}$ less
$37$ for the number of linear restrictions involved,  36 from the
frequency mimic assertions of type i, plus 1 from the unit-sum
restriction.)

There are two ways to complete the bases for
understanding this family of distributions: by a direct
characterisation of the polytope via de Finetti's fundamental
theorem of prevision, and then exploring the space of complete
cohering distributions using robust programming software such as the GAMS package \citep{Brooke}%(Brooke et al, 2003) 
or by examining an array of interpretable
reference distributions within the class of all agreeing FMD's. The
graphical results we display in this
article follow the latter tack. 

This article systematically examines the coherent implications of
assertions such as those of the apiarist for prior distributions
over the proportion of successes to be observed in the total population.  For reasons we shall see, the practical
relevance of the results are to finite population problems as opposed to sequential experimentation to an unspecified extent. In Section
\ref{sect:prel} we describe two algebraic features of the computational
algorithms used to generate the particular cases reported in this
article; and we review one background result pertinent to the uniform distribution which is not
widely known. A general nonparametric specification of FMDs is
presented in Section \ref{sect:fmd}, and an interactive file of subroutines to
display them is made publicly available. This allows the reader to
examine any cases of interest whatever, and to use the results for
application. Some interesting numerical examples are displayed and
discussed in Section \ref{sect:numericalexamples}. These results continue and extend the
specific numerical assessments portrayed in our introductory apiary
example. Section \ref{sect:issue} addresses structural issues surrounding the
extendibility of FMDs. It presents three new theorems that
together specify rather precisely the applicable relevance of
frequency mimicking distributions.  Proofs are presented in sparse
algebraic form in the text, along with computed examples. Appendices
are used to develop more extensive detail as well as to portray a helpful
geometrical exposition. Implications of the analysis extend to the limits of infinitely extendible but only
finitely additive distributions. These limiting distributions are
improper and feature ``adherent'' or ``agglutinated'' masses. This
is a technical topic that amused
\cite{deFinetti1949,deFinetti55}
% de Finetti (1949, 1955, 1972)
immensely, and is relevant to his leading work on applications of
conditional probabilities when the conditioning event is assessed
with probability zero.  See \cite{Cifarelli96}.%Cifarelli and Regazzini (1996)
We shall merely report and discuss these results in the text here, and make available  ``supplementary materials'' in Appendix 4 to provide algebraic detail.
 To refresh the reader with this feature that will be
relevant to the constructions in
the article, we present the following definition. 
\begin{definition}\label{def:first}
%\noindent {\bf Definition:} 
A finitely additive probability
distribution for $X $ over $[0,1]$ is said to have adherent (or
agglutinated) masses of size $p$ and $1-p$ at $1 $ and $0$,
respectively, if $P(X = x)  = 0$ for any $x$, yet for {\it any} numbers $a$
and $b$ for which $\; 0 < a < b < 1, \ P[X \in [a,b]] = 0, \ $
while $\ P[X \in (b, 1)] = p \ $
and $P[X \in (0, a)] = 1-p$. \ \ $\S$  
\end{definition}	
Properties of distributions with adherent masses may appear unusual,
because such distributions are only finitely additive, not countably additive. 
%The distribution function is {\it not} lower semi-continuous, which is a property honoured by all countably additive distributions. 
The image of adherence is that the total probability of 1
does not attach itself to any open intervals that are separated from
0 and 1, since the probability that $X$ lies in any interval
interior to $(0, 1)$ is zero. Yet the entire probability of $1$
{\it adheres} to the endpoints of the unit-interval without amassing on
the points $0$ or $1$ themselves. Applications in statistics
were presented by 
\citet*{Kadane86}.
%Kadane, Schervish and Seidenfeld (1986)
The  
article of \cite{Bingham10}
%Bingham (2010)
discusses the historical development of
the analytic assessment of finitely additive measures.  A
contribution of the present article to the extensive literature on
this topic is to show how to construct examples of such
distributions from sequential extensions of purely finite
applications with sensible properties.  Of course 
the definition of adherent masses can be embellished to allow the
points of agglutination to occur anywhere within the interval $[0,1]$. 

Our reported results can be understood as a completion of the
nonparametric investigations of 
\cite{Hill88}
%Hill (1988)
who introduced the
generation of so-called $A_n$ and $H_n$ distributions for continuous
measurements;  and they share the finitist attitude toward
nonparametric inference described more extensively in \cite{Hill89}.
%Hill (1989).
The $A_n$ and $H_n$ distributions for continuous data are
characterised by posterior distributions asserting uniform
probabilities that the next measurement will lie in each of the
various intervals defined by the order statistics of the
conditioning observations.  The first important application appeared in the
article of \cite{Berliner88}.
% Berliner and Hill (1988). 
The so-called $A_n$ distributions
pertain to a measurement context that {\it does not allow ties},
while the $H_n$ distributions {\it do allow ties} among the
conditioning observations. In the context of events regarded
exchangeably that we study here, {\it ties are required}, since the
posterior probability for the next measurement equals the frequency
of occurrence of successes among the conditioning events. The sum
$S_N$ indicates the number of the events observed to be tied at $1$,
while $(N-S_N)$ indicates the number of them tied at $0$.  More
detailed commentary on the relationship of this work to Hill's
analysis appears in the technical report of 
\citet[pp. 52--55]{Lad93}.
%Lad et al (1993, pp. 52-55).

The larger context of the present article is the widely studied {\it predictive characterisation} of exchangeable distributions that reduce the families of supporting distributions to either a parametric or a specifiable nonparametric class.  Relevant literature is reviewed in \cite{Fortini12}.% Fortini and Petrone (2012).
 While our purely finite results are completely nonparametric, the limiting distribution for the finite family, derived in Section \ref{sect:5.5limitfmd}, provides a recognisable unifying parametric envelope for the FMD family.  Prior to our exact characterization of FMD's, the nearest related results were limited to approximations and asymptotics.  These can be found in 
 \citet*{Berti09} and  \citet*{Cifarelli16}.
% Berti, Crimaldi,  Pratelli, and Rigo (2009), and Cifarelli, Dolera, and Regazzini (2016).

\section{\ Preliminary technical review}
\label{sect:prel}
Before formalising our problem for analysis, we first review
two computational features of the relations between the probability
mass function vector {\bf q}$_{N+2}$ and the conditional predictive probability
vector {\bf p}$_{N+1}$. Then we recall an algebraic result about
finite distributions that mimic frequencies over the {\it entire domain} of positive frequencies.  
\subsection{Inversion equations}\label{sect:inversion}
Well known in the context of exchangeability, standard formulas can
be used to compute conditional probability assertions in the vector
{\bf p}$_{N+1}$ from an unconditional probability mass function {\bf
	q}$_{N+2}$ when all components $q_{a,N+1}$ are strictly positive.  This is
achieved in this ``standard case'' by the nonlinear equations
\begin{equation}
p_{a,N} \ = \ \frac {(a+1) \ q_{a+1,N+1}} {(a+1)
	\ q_{a+1,N+1} \ + \ (N+1-a) \ q_{a,N+1}}\;, \ \ \ {\rm for}
\ a = 0, 1, ..., N\;. \  \label{eq:directeq}
\end{equation}
\noindent This result derives easily from the fact that $P[E_{N+1}|(S_N=a)]
= P[E_{N+1}(S_N=a)]/P(S_N=a)$ via the exchangeability structure and
algebraic simplification.  An article of \cite{deFinetti1952} %de Finetti (1952)
contains an exhaustive analysis of the ``degenerate case'' when any
components
of {\bf q}$_{N+2}$ are allowed to equal zero.  The distribution of any subsequence of events from ${\bf E}_{N+1}$ is mixture-hypergeometric given the sum $S_{N+1}$, with a mixing function specified by some pmf vector  ${\bf q}_{N+2}$ in the unit-simplex ${\bf S}^{N+1}$.  This specification characterises the family of all finite exchangeable distributions over \vspace{.06cm} ${\bf E}_{N+1}$.

However, it is not widely recognised that these equations
(\ref{eq:directeq}) are invertible, yielding {\bf q}$_{N+2}$ as a
nonlinear function of {\bf p}$_{N+1}$ via a recursive formula derived in \citet*[p. 198]{Lad95}:
% Lad, Deely and Piesse (1995, p. 198):
\begin{equation}
q_{0,N+1} \ = \ \{1 + \displaystyle\sum_{a =1}^{N+1}
\binom{N+1}{a}
%\left({N+1\atopa+1}\right)
\prod_{i=0}^{a-1}\ \frac{p_{i,N}}{1-p_{i,N}}\}^{-1}= \ \  , \ \ {\rm and}
\label{eq:intrigue}
\end{equation}

$$\ \ \ \ \ \ \ \ \ \ \ \ \ \ \ \ \  q_{a,N+1} \ = \ 
\binom{N+1}{a}
% \left({N+1\atop a}\right)
\prod_{i=0}^{a-1}\frac{p_{i,N}}{1-p_{i,N}}\ \ \ q_{0,N+1} \; \, \ \
\ \ \ \ {\rm for}\ \ \ \, a = 1, .. . ,
N+1\ .\ \ \ \ \ \ \ \ \ \ \ \ \ $$\\
% $$ \noindent for $
\noindent The unconditional probabilities in the vector {\bf
	q}$_{N+2}$ are computed via products of increasing numbers of odds
ratios corresponding to the conditional probabilities composing {\bf
	p}$_{N+1}$. Their recursive specification begins with the odds ratio specifying $q_{1,N+1}$ and continues sequentially by multiplying successive odds ratios.  Normalisation is achieved through the determination of $q_{0,N+1}$ after the product odds ratio is completed for $q_{N+1,N+1}$. 
Equations (\ref{eq:intrigue}) will be used in several
derivations for the analysis reported in the present article\vspace{.06cm}.  

This inversion result is not merely a computational oddity.  It actually specifies a characterisation of exchangeable distributions over ${\bf E}_{N+1}$ that is equivalent to their well-known characterisation 
in terms of coherent pmf's for $S_{N+1}$ within the unit-simplex.  Any pmf ${\bf
	q}_{N+2}$ within the unit-simplex ${\bf S}^{N+1}$ for the sum $S_{N+1}$ yields a coherent predictive probability vector
${\bf p}_{N+1}$ within the $(N+1)$-dimensional unit-cube via equations (\ref{eq:directeq});  conversely, any predictive probability vector 
${\bf p}_{N+1}$ within the $(N+1)$-dimensional unit-cube yields a coherent pmf ${\bf
	q}_{N+2}$ for $S_{N+1}$ within the unit-simplex ${\bf S}^{N+1}$ via equations (\ref{eq:intrigue}).  An interesting feature of this characterisation of coherent finite exchangeable distributions is that it does not require any recourse to a parametric mixture distribution, as required in the literature summarised by  \citet[Section 3]{Fortini16}.  The conundrums they mention concerning solutions of functional equations are limited to the more restricted distributions they consider which respect {\it complete}\vspace{.06cm} additivity. 
%{\bf Pippo, we need to get this reference to the Fortini/Petrone article into the ``cite'' syntax of bibtex for references. DONE}
%{\bf Now I think these next two results are very worthwhile to publish with this article.  However, we should evaluate this once we see how long this newly edited version will be.}

The inverse transformation (\ref{eq:intrigue}) from {\bf p}$_{N+1}$ to {\bf q}$_{N+2}$   can also be expressed directly via
\begin{equation}\label{eq:genbin}
q_{a,N+1} \ = \ K\; 
\binom{N+1}{a}
%{N+1\choose a}
\prod_{i=0}^{a-1} p_{i,N}\prod_{i=a}^{N}(1-p_{i,N})\ \; \,
\ {\rm for}\ \ \ \, a = 0, .. . ,
N+1\ ,
\end{equation}
where by convention, $\prod_{i=0}^{-1}
p_{i,N}=\prod_{i=N+1}^{N}(1-p_{i,N})=1$, and $K$ is the normalising constant 
\[
K= \frac{1}{\sum_{a=0}^{N+1} {N+1\choose a}\prod_{i=0}^{a-1} p_{i,N}\prod_{i=a}^{N}(1-p_{i,N})}\,\,.
\] This form  identifies the family of all exchangeable distributions as a generalisation of the Binomial distributions, for which every value of $p_{i,N}$ is constant at some value of $p \in(0,1)$. 

Finally, the equations for $q_{a,N+1}$ can be expressed
recursively in still another way as well. Specifically,
\begin{equation}
q_{a+1,N+1} \ = \ \left(\frac{N+1-a}{a+1}\right) \ \left(\frac{p_{a,N}}{1 - p_{a,N}}\right) \ q_{a,N+1} \ , \ \ \ \ {\rm for}\  a = 0, .. . ,N.
\label{eq:recursiveform}
\end{equation}

\noindent This recursive form identifies {\it linear} conditions
among the components of {\bf q}$_{N+2}$ arising from the inversion
equations.
\subsection{The inversion of reduction probabilities}\label{sect:theinversion}
Secondly, in the context of exchangeability, the reduction of a
probability mass function over the sum of $N+1$ events (represented
by {\bf q}$_{N+2}$) to the cohering distribution over the sum of the
first $N$ of them, {\bf q}$_{N+1}$, also follows well-known
formulas:
%\begin{equation}
%q_{a,N} \ = \ ({N \choose a}/{N+1 \choose a}) \ q_{a,N+1} \ + \ ({N
%\choose a}/{N+1 \choose a+1}) \ q_{a+1,N+1}\ \ 
%\hspace*{4cm}{\rm for} \ a = 0,
%1, ..., N\ . \hspace{3cm}  \label{eq:usualreduction}
%\end{equation}
%\begin{eqnarray}
%q_{a,N} \ = \ ({N \choose a}/{N+1 \choose a}) \ q_{a,N+1} \ + \ ({N
%	\choose a}/{N+1 \choose a+1}) \ q_{a+1,N+1}\ \ \\ 
%\hspace*{4cm}{\rm for} \ a = 0,\nonumber
%1, ..., N\ . \hspace{3cm}  \label{eq:usualreduction}
%\end{eqnarray}
\begin{equation}
q_{a,N} \ = \ \frac{{N \choose a}}{{N+1 \choose a}} \ q_{a,N+1} \ + \ \frac{{N
		\choose a}}{{N+1 \choose a+1}} \ q_{a+1,N+1},\ \ {\rm for} \ a = 0,
1, ..., N\ .  \label{eq:usualreduction}
\end{equation}

\noindent These reduction equations derive from applying
exchangeability conditions to the fact that $P(S_N=a) \ = \
P[(S_N=a)\tilde{E}_{N+1}] + P[(S_N=a)E_{N+1}]$, where $\tilde{E}_{N+1} \equiv 1 - E_{N+1}$. 

The corresponding reduction formulas generating
lower-order {\it conditional} probabilities such as $p_{a,N-1}$ from
$p_{a,N}$ and $p_{a+1,N}$ are not usually considered.   These resolve
to the equations
%\begin{equation}
%p_{a,N-1} \ \ = \ \ p_{a,N} \ / \ (1 - p_{a+1,N} + p_{a,N}) \ ,
%\ \ \ \ {\rm for} \ \ a \ = \ 0, 1, ..., N-1\ , \label{eq:reduction}
%\end{equation}
\begin{equation}
p_{a,N-1} \ \ = \ \ \frac{p_{a,N} }{1 - p_{a+1,N} + p_{a,N}}\ ,
\ \ \ \ {\rm for} \ \ a \ = \ 0, 1, ..., N-1\ , \label{eq:reduction}
\end{equation}

\noindent which derive from applying equations (\ref{eq:intrigue})
to those of (\ref{eq:usualreduction}).  See 
\citet[pp.
199, 204]{Lad95}.
%Lad et al (1995, pp. 199, 204).
The reduction equation (\ref{eq:reduction}) will be used
to prove the reduction and extension theorems to be discussed in
Section \ref{sect:issue}.   When exchangeable {\it extensions} are addressed, it will be used in a form relevant to
the next larger value of $N$:
%
%%\begin{align*}
%\noindent \hspace*{3mm}$\ \ p_{a,N} \ \ = \ \ p_{a,N+1} \ / \ (1 - p_{a+1,N+1} + p_{a,N+1}) \ ,
%\ \ \ \ ${\rm for} $\ \ a \ = \ 0, 1, ..., N\ . \ \ \ \ \ \vspace{.15cm} \ \ \  (5')$ 
%%\nonumber
%%\end{align*}
\[
p_{a,N} \ \ = \ \ \frac{p_{a,N+1}}{
	1 - p_{a+1,N+1} + p_{a,N+1}} \ ,
\ \ \ \ {\rm for} \ \ a \ = \ 0, 1, ..., N\ .  \tag{\ref{eq:reduction}$'$}
\]

It is worth remarking to conclude this Section that equation (\ref{eq:usualreduction}) is 
a special case of a general reduction equation that would reduce {\bf q}$_{N+2}$ to, say, 
{\bf q}$_{M+1}$ for any value of $M \leq N$:

\begin{equation}
q_{a,M}\ \ =\ \ \sum_{A=a}^{N+1-(M-a)}\; \frac{{A \choose a}{N+1-A \choose M-a}}{{N+1\choose M}}\ \ q_{A,N+1}\ \ , \ \ \ \ \ {\rm for} \ a=0,\ldots,M\ \ \ .  
\label{eq:generalizedusualreduction}
\end{equation}

\noindent This stems from the fact that partial sums of exchangeable sequences are distributed as mixture hypergeometric with respect to the sum of the entire sequence.

\subsection{FMD's and the improper uniform  distribution}\label{sect:thm1}
Predictive probabilities for exchangeable sequences based on the
improper uniform prior {\it always} mimic positive conditioning frequencies,
for {\it any} observed frequency within $(0,1)$. An algebraic
analysis of finite agreements with positive conditioning frequencies
by \citet[pp. 207--208]{Lad95} 
%Lad et al (1995, pp. 207-208) 
has yielded a complete
explicit result in every finite context for which {\bf q}$_{N+2}$ is
strictly positive. It places the infinitely extendible improper uniform 
mixture distribution within the context of the class of positive
frequency mimicking distributions for any finite size of $N+1$. We
state it here as Theorem \ref{thm:first}, followed by a brief
discussion.  \vspace{1.5mm} 

%\noindent {\bf Theorem 1:} Suppose $N+1$ events are regarded
%exchangeably. If a 
%\todo{We require this, but later we write $q_{0,N+1}\in[0,1]$. Maybe here we cannot require this positivity}{{\it strictly positive}} probability mass function
%over the sum of these events implies conditional probabilities that
%agree with all positive conditioning frequencies, $p_{a,N} = a/N$
%for the values of $a = 1, 2, ..., N-1$, then the mass function {\bf
%q}$_{N+2}$ must subscribe to
%the restrictions that \\

\begin{theorem}\label{thm:first}
 Suppose $N+1$ events are regarded
exchangeably. If $q_{a,N}>0$  over the values of
$a = 1, 2, ..., N-1$, and the associated conditional probabilities for $E_{N+1}$ agree with positive conditioning frequencies, $p_{a,N} = a/N$, then the pmf vector {\bf
	q}$_{N+2}$ must subscribe to
the restrictions that \vspace{.1cm} \\
i.$\;$) \; $q_{0,N+1} \in [0,1)$\ \ $ {\rm and} \ \ 0 \; \leq
\; q_{1,N+1} \; \leq \;   \frac{(1 - q_{0,N+1})}{2 \; [N/(N+1)]
	\ H(N)} \; \equiv \; B_{N+1}(q_{0,N+1})\;$,\\

\hspace*{0.25cm}$ \ \ \ \ \ $where $H(N)\ = \ \Sigma_{a=1}^{ \; N} \
a^{-1}$,  the $N^{th}$ harmonic sum\;;

\noindent ii.$\; $) \ \ $q_{a,N+1} \ = \ (1/a) \ [N/(N-a+1)] \
q_{1,N+1}$ \ \ \ for $a = 2, ..., N$; \ \ \  and \vspace{.2cm}

\noindent iii.) \ \ $q_{N+1,N+1} \ = \ 1 - q_{0,N+1} - 2q_{1,N+1} \
[N/(N+1)] \ H(N)\; $. \vspace{.2cm}  \ \ \ \ \  %$\diamond$
\end{theorem}
%\noindent {\bf Theorem 1:} Suppose $N+1$ events are regarded
%exchangeably. If $q_{a,N}>0$  over the values of
%$a = 1, 2, ..., N-1$, and the associated conditional probabilities for $E_{N+1}$ agree with positive conditioning frequencies, $p_{a,N} = a/N$, then the pmf vector {\bf
%q}$_{N+2}$ must subscribe to
%the restrictions that \vspace{.1cm} \\
% i.$\;$) \; $q_{0,N+1} \in [0,1)$\ \ $ {\rm and} \ \ 0 \; \leq
%\; q_{1,N+1} \; \leq \;   \frac{(1 - q_{0,N+1})}{2 \; [N/(N+1)]
%\ H(N)} \; \equiv \; B_{N+1}(q_{0,N+1})\;$,\\
%
%\hspace*{0.25cm}$ \ \ \ \ \ $where $H(N)\ = \ \Sigma_{a=1}^{ \; N} \
%a^{-1}$,  the $N^{th}$ harmonic sum\;;
%
%\noindent ii.$\; $) \ \ $q_{a,N+1} \ = \ (1/a) \ [N/(N-a+1)] \
%q_{1,N+1}$ \ \ \ for $a = 2, ..., N$; \ \ \  and \vspace{.2cm}
%
%\noindent iii.) \ \ $q_{N+1,N+1} \ = \ 1 - q_{0,N+1} - 2q_{1,N+1} \
%[N/(N+1)] \ H(N)\; $. \vspace{.2cm}  \ \ \ \ \  $\diamond$
%
%% where $\Psi(^.)$ is the psi-function, defined for $x \geq 2$ by
%% $\Psi(x) = -\gamma \ + \ \Sigma_{a=1}^{x-1} \ a^{-1}$,
%% and $\gamma$ is Euler's
%% constant, equal to $0.5772156649...$ .  \ \ \ \ \  $\diamond$\\
%
%% The psi-function is the discrete analogue of the logarithm function
%% in its integral representation:  $log(x) = \int_1^x t^{-1} dt$. See,
%% for example, Abramowitz and Stegun (1964, p. 258).  \\

\noindent {\bf Comments:}  Since the vector {\bf q}$_{N+2}$ lies
within the $(N+1)$-dimensional unit-simplex, the presumed $(N-1)$
assertions that conditional probabilities equal {\it any} positive
conditioning frequencies leave only two dimensions of freedom in
specifying {\bf q}$_{N+2}$. Identifying the free variables as
$q_{0,N+1}$ and $q_{1,N+1}$, their restricted triangular
2-D region (specified in restriction \emph{i} of  Theorem \ref{thm:first}) diminishes to the 1-dimensional unit-interval as the size of $N$ increases. Although Theorem \ref{thm:first} allows
$q_{0,N+1}$ to take any value within $[0,1)$, it presses $q_{1,N+1}$
toward a limit of $0$ as $N$ increases, because the harmonic series
diverges. Thus, according to statement \emph{ii} of  Theorem \ref{thm:first}, all other
components $q_{a,N+1}$ converge to $0$ as well for values of $a = 2, ...,
N$.  The limiting property of the sum $\Sigma_{a=1}^{\; N} \
q_{a,N+1}$ requires further analysis, since the number of summands in the series is unlimited as $N$ increases. The further analysis
provides that this entire series converges to $0$. In light of this result
which we now address, the value of
$q_{N+1,N+1}$ converges to $1-q_{0,N+1}$. \vspace{.2cm}

Statement \emph{iii} of Theorem \ref{thm:first} represents the unitary
summation constraint on components of {\bf q}$_{N+2}$.  The limiting
behaviour of the second subtracted term in this equation appears
problematic because it is the product of $q_{1,N+1}$, which is
converging to 0, and a coefficient that increases without bound --
the harmonic series.  This second subtracted term equals
$\Sigma_{a=1}^{\ N \ }\ q_{a,N+1}$, which is the sum of an unbounded {\it
number} of terms each of which converges to $0$. That this entire
sum converges to $0$ merits a second theorem of its own, which we
shall now formulate and prove. What this result implies is that the
limit of the finitely additive distributions for the frequency of occurrences is concentrated only near the endpoints $0$ and $1$.\vspace{.15cm} 
\begin{theorem}
$\lim_{N\rightarrow\infty}P(1 \leq S_{N+1} \leq N) 
\ =  \ 0 $ \  under the conditions of Theorem \ref{thm:first}.
\end{theorem}

% = \Sigma_{a=1}^{ \; N} \ q_{a,N+1} \ 

%\noindent {\bf Theorem 2:} Under the conditions of Theorem 1 which
%we presume to hold for any value of $N$,
%\ \ 
%$P(1 \leq S_{N+1} \leq N) \ = \ \Sigma_{a=1}^{ \; N} \ q_{a,N+1} \ \
%\rightarrow \ 0 \ \ as \ \ N \ \rightarrow \ \ \infty \ . \ \ \ \ \
%\diamond$ 

\begin{proof} Let $M$ be any fixed integer, $M < N+1$. Now using
equation (\ref{eq:generalizedusualreduction}) to reduce the mass function for $S_{N+1}$
to its implied mass function over the sum of $M$ events yields
the result that 
\[ P(S_M = a) \ \leq \ ^MC_a \ [(N+2-M)/(N+1)] \ q_{1,N+1} \ \ {\rm for \ any} \ a = 1,...,M-1\; .
\nonumber
\]
\noindent Algebraic details of this derivation are presented  expansively in
the technical report of 
\citet[pp 19--20]{Lad93}.
%Lad et al (1993, pp 19-20). 
 Since the value of $a$ appears only in the combinatoric
expression $^MC_a \equiv M!/[a!(M-a)!]$, the sum of these probabilities is bounded:
\begin{equation}
\sum_{\; a=1}^{M-1}\; P(S_M = a) \ = \ \sum_{\ a=1}^{M-1} \ q_{a,M} \ \leq \ 2^M \ [(N+2-M)/(N+1)] \ q_{1,N+1} \ \ ,
\label{eq:sumbound}
\end{equation}
\noindent because $\sum_{a=1}^{ M-1}\;^MC_a \; < \; 2^M$. Thus, the
sum $\sum_{\ a=1}^{M-1}\; P(S_M = a)$  converges to $0$ along with $q_{1,N+1}$ as $N$
increases. 
Morever, since Equations (\ref{eq:usualreduction}) and (\ref{eq:reduction}) imply that the conditions of Theorem \ref{thm:first} apply to sequences of any size $M$ as well as $N$ increases, 
it follows that  $\lim_{N\rightarrow\infty}P(1 \leq S_{N+1} \leq N) 
\ =  \ 0.$ \ \ $\qed $ 
\end{proof}

The improper prior distribution that is uniformly zero over the open interval $(0,1)$, whose mixture supports the conditions of
Theorem \ref{thm:first} for every value of $N \geq 1$, can be understood as the limit of a sequence of finitely
additive mixing distributions that are all FMD's on the entire open unit-interval
$(0,1)$.  The present
article extends the analysis of Theorem \ref{thm:first} to conditions when the
predictive probabilities are presumed only to mimic frequencies {\it
only} over a specified rational interval properly {\it within}
$(0,1)$.  For finite sizes of $N$ the associated pmf's are found to be well behaved and appealing for use in applied problems. Herein, we shall also show how to construct a whole family of limiting distributions (FMD's over restricted subdomains) that exhibit
agglutinated masses which merely adhere to the endpoints, $0$ and $1$, with
recognisable degrees of stickiness.  These variations have their
sources in the size and position of the frequency mimicking
subdomains within $(0,1)$. 

\section{\ FMD space, reference FMDs, and computational software}\label{sect:fmd}

In the course of this discussion we shall display and assess a few
specific reference distributions that all satisfy the following
shared properties of frequency mimicking distributions, for various
sizes of N:  
\begin{itemize}
\item {$P(E_{N+1}|S_N=a)= a/N$ \ \ if $a \ \in [a_1, a_2]$ \ ; }
\item {Conditional probabilities
$P(E_{N+1}|S_N=a)$ are monotone non-decreasing in $a$;}
\item{$P(E_{N+1}|\;S_N=\;0) \geq p_L(N) > 0$, \ where $p_L(N) \leq a_1/N$} \ ; \ and
\item {$P(E_{N+1}|S_N=N) \leq p_U(N) < 1$, \ where $p_U(N) \geq a_2/N \ $.}
\end{itemize}

\noindent To denote assertions in the form of these general specifications, we
shall refer to ``asserting 
$PaN[a_1, a_2, p_L(N), p_U(N)]$.'' When
referring to such  assertions for a lower value $n$ of $N$, we use the
notation $Pan[^., ^.,^., ^.]$.  {\it Notice explicitly} that the pair $(\frac{a_1}{N},
\frac{a_2}{N})$ denote the endpoints of the subdomain of presumed
frequency mimicking, whereas $p_L(N)$ and $p_U(N)$ denote lower and
upper bounds on predictive probabilities when the conditioning
frequencies equal $0$ and $1$, respectively.  The upper bound on $p_L(N)$ follows from the non-decreasing feature of the $p_{a,N}$;  similarly for the lower bound on $p_U(N)$. 

As mentioned in our introduction, these conditions do not identify a
unique cohering probability distribution, but rather a whole convex
space of distributions characterised via de Finetti's FTP.  We can
get a feel for the range of distributions composing this space by
studying four specific reference distributions near the extremes and
at the heart of this space, as displayed in Figure
\ref{fig:panpicture}. These differ only in how they assess
conditional probabilities $P(E_{N+1}|S_N = a)$ when $a/N$ lies
outside the interval of unadjusted relative frequencies,
$[\frac{a_1}{N}, \frac{a_2}{N}]$. The distributions we examine are
aptly named according to their distinctive properties: Linear,
Quartic, Weak Extreme and Strong Extreme.  In studying their
algebraic descriptions which follow, refer to the labeled example
functions displayed in Figure \ref{fig:panpicture}. {\it Be aware} that these displayed
functions have been produced to appear continuous for display
purposes relevant to {\it any} size of $N$. In fact, based on the
specification of a specific finite $N$, they are discrete functions of
$a/N$ for $a = 0, 1, 2, ..., N$.

\begin{figure}[!th]
\begin{center}
\includegraphics[width=1\linewidth]{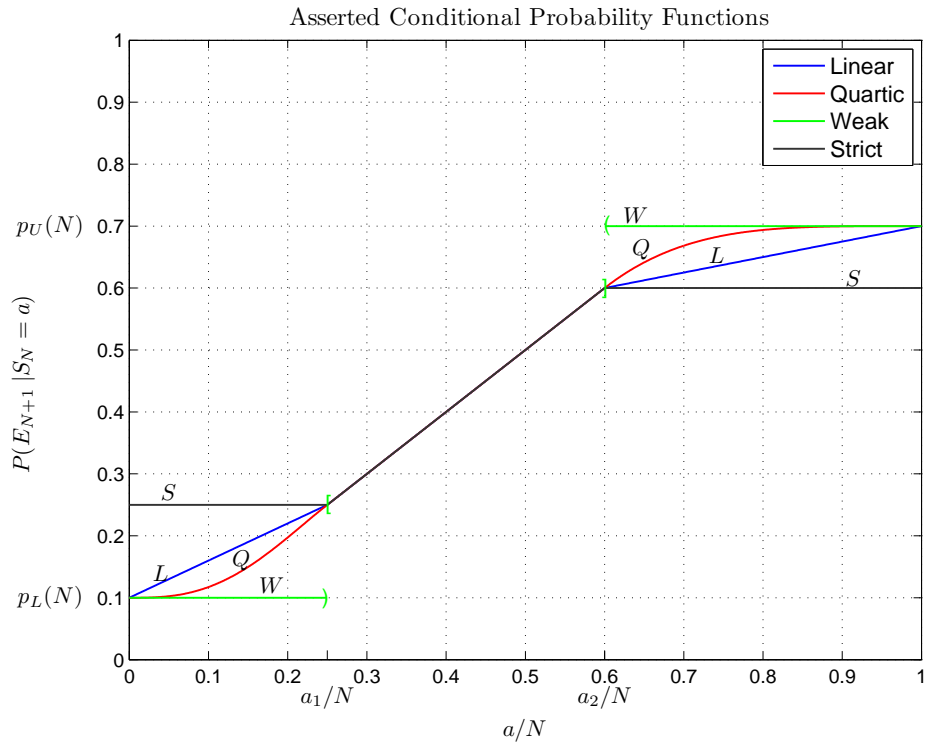}\\
\caption{Linear, Quartic, Weak and Strict conditional probability functions which agree with conditioning frequencies over the restricted subdomain between 
$\frac{a_1}{N} = .25$ and $\frac{a_2}{N} = .60$.
Notice the positions of $p_L(N) = .10$, and
$p_U(N) = .70$ on the ordinate.  }
 \label{fig:panpicture}
\end{center}
\end{figure}
\begin{enumerate}
{\item {\bf Linear:}

The values of $P(E_{N+1}|S_N=a)$ increase linearly over
$\frac{a}{N}$ within $[0, \frac{a_1}{N}]$ and $[\frac{a_2}{N},
1$].}

{\item {\bf Quartic:}

The values of $P(E_{N+1}|S_N=a)$ are specified by a polynomial
increasing quartic function $Q_L(\frac{a}{N})$ for values of \
$a/N \in [0,\frac{a_1}{N}]$. This lower quartic function
$Q_L(^.)$ is determined to satisfy the end-point and derivative
conditions on $Q'_L(^.)$ that
\[
\begin{array}{cccc}
 Q_L(0)=p_L(N), & Q_L(\frac{a_1}{N})=\frac{a_1}{N},& Q_L'(0)=0,
\ {\rm and} & Q_L'(\frac{a_1}{N})=1 \ \ \ .
\end{array}
\]
Above the upper end of the frequency mimicking interval, the
values of $P(E_{N+1}|S_N=a)$ are defined by a polynomial
increasing quartic function  $Q_U(\frac{a}{N})$ for values of \
$a/N \in [\frac{a_2}{N},1]$. This upper quartic function is
determined by the conditions on it and on its derivative
function, $Q'_U(^.)$, that
\[
\begin{array}{cccc}
Q_U(\frac{a_2}{N})=\frac{a_2}{N}, & Q_U(1)=p_U(N), & Q_U'(\frac{a_2}{N})=1, \ {\rm and} & Q_U'(1)=0 \
\ \ .
\end{array}
\] }
{\item {\bf Weak extreme:}
\[
P(E_{N+1}|S_N=a)=p_L(N) \ \ \mbox{ if } \ a/N \ \in [0,\frac{a_1}{N})
\ \ \ \ \ \ , \ {\rm and} \]
\[
P(E_{N+1}|S_N=a)=p_U(N) \ \ \mbox{ if } \ a/N \ \in
(\frac{a_2}{N},1] \ \ \ \ \ \ . \ \ \ \ \ \  \]}
{\item {\bf Strict
extreme:}
\[
P(E_{N+1}|S_N=a)=\frac{a_1}{N} \ \ \mbox{ if } \ a/N \ \in [0,\frac{a_1}{N})
\ \ \ \ \ \ \ \ , \ \ \ {\rm and} \]
\[
P(E_{N+1}|S_N=a)=\frac{a_2}{N} \ \  \mbox{ if } \ a/N \ \in (\frac{a_2}{N},1]
\ \ \ \ \ \ . \ \ \ \ \ \  \]}
\end{enumerate}
The weak and strong
functions portray opposing attitudes toward the strictness of the
proclaimed frequency mimicking interval:  the strict function never
allows conditional probabilities below the lower endpoint value of
this interval nor above the upper endpoint;  the weak function
proclaims conditional probabilities equal to the proclaimed minimum
and maximum valuations $p_L(N)$ and $p_U(N)$ as soon as the
conditioning frequency is observed outside the FMD interval. The
linear functions bisect these bounding regions, and agree with the
lower and upper limit points on conditional probabilities when the
conditioning sum of successes equals $0$ or $N$. The quartic
functions also agree at these endpoints, but the conditions on their
derivatives ensure that their approaches to the endpoints and to the
agreeable frequency
region are smooth. 

The software we have designed to generate the graphical results
displayed in the next Section is freely available, using the link \url{http://www.unipa.it/sanfilippo/mimic}. User friendly, there are two version: one is based on Shiny R and one it is based on  MATLAB code.
%\todo[inline]{Replace the prvious part by something like
%``The software we have designed to generate the graphical results
%displayed in the next Section is freely available, using the link \url{http://www.unipa.it/sanfilippo/mimic}. User friendly, there are two version: one it is based on Shiny R and one it is based on  MATLAB code.}
The user need only enter the
sizes of $N, \frac{a_1}{N}, \frac{a_2}{N}, p_L(N)$ and $p_U(N)$ in
prompting boxes, and graphical displays of the associated linear,
quadratic, weak and strong extreme mass functions are produced. In addition to the function values of $p_{a,N}$, the associated mass function values
$q_{a,N}$ are produced as well, computed via the  equations
(\ref{eq:intrigue}). We shall view examples of these computations in the next
Section. 

\section {\ Numerical examples}\label{sect:numericalexamples}

Figures~\ref{fig:pdfpicture} and \ref{fig:pdfpicturebig} display probability mass
functions {\bf q}$_{N+2}$ pertinent to our apiary example, computed
for the four reference FMDs over a limited domain that we now have formalised. In evaluating these Figures, you should be aware of one
detailed feature of 
\begin{figure}[!ht]
\begin{center}
\includegraphics[width=1\linewidth]{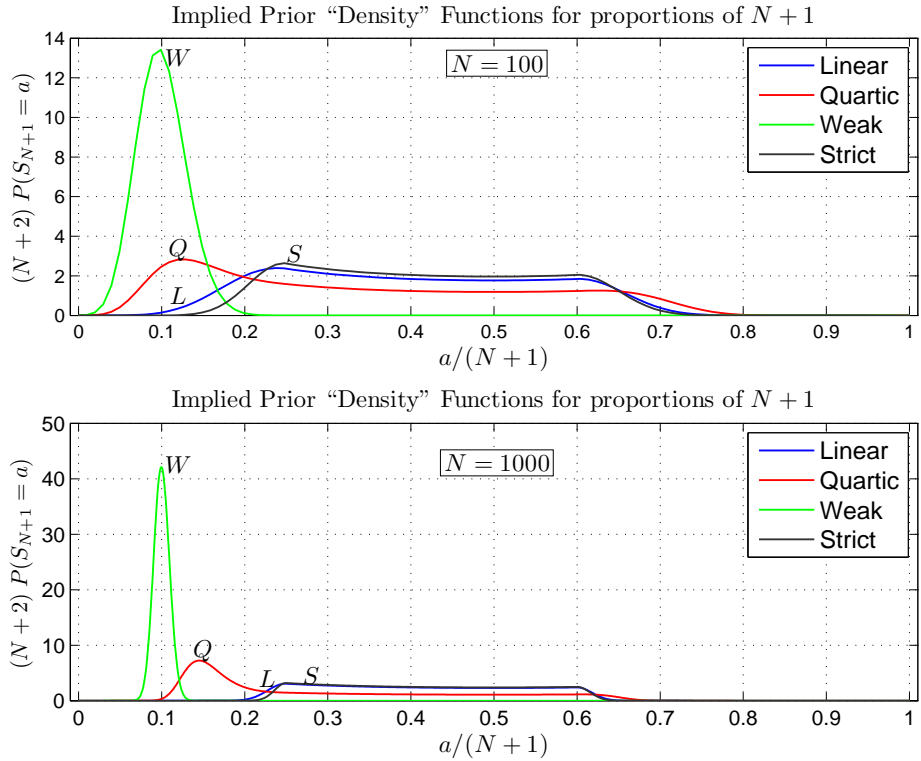}
\caption{The upper graph displays
 functions proportional to the pmf functions labeled L, Q, W and S when $N =
 100$, while the lower graph displays these same functions implied when
 $N = 1000$.  These correspond to the $PaN$ assertions displayed in Figure \ref{fig:panpicture}.}

\label{fig:pdfpicture}
\end{center}
\end{figure}
their construction. Although they appear to be
continuous, the functions presented are actually mass functions on a
grid of points within the unit-interval, with positive masses only
on the discrete domain of points appropriate to the size of $(N+1)$.
These mass points have been transformed into smoothed normed
histograms in the Figures.  This convention will allow
us to display and to distinguish mass functions associated with
different sizes of $N$ on the same graph during the course of our
discussion in Section \ref{sect:issue}. Without recourse to this convention, the
varying scales of {\bf q}$_{N+2}$ functions would preclude visual comparison.  Specifically, the mass functions have
been computed as normed ``density histograms'' as described by \cite[pp. 114--115]{Martinez02}.
%Martinez and Martinez (2002, 114-115). 
In brief, if a histogram is
constructed with bin widths $h$, each histogram frequency $f_i$ is
normalised to $d_i = f_i/h$, assuring that the displayed histogram
``density'' integrates to $1$. When any one of our mass functions
based on $N+1$ events is constructed, it involves $N+2$ bins.  Thus,
each bin width in the unit-interval equals $(N+2)^{-1}$.  As a
result, ordinate labels on these functions read ``$(N+2)\ P(S_{N+1}=a)$'',
and the title on the function in various Figures refer to them as
``Density'' functions, in quotation marks.

Figure \ref{fig:pdfpicture} displays probability mass functions for
$S_{N+1}$ that are implied by the L, Q, W and S specifications under
the frequency mimicking assertions for the apiary example shown in
Figure \ref{fig:panpicture}.  The top bank of the Figure applies to
values of $N$ equal to $100$, while the lower bank is constructed
for $N$ equal to $1000$. The values of $p_L(N)$ and $p_U(N)$ are
$.1$ and $.7$, respectively, for both of these examples.  It is
apparent that each of the four pmf ``density'' types contracts and
sharpens as $N$ increases. (Notice the different scales on the
ordinate axes of the two graphs.)  The ``density'' for the Weak
extreme appears unusual relative to those of the other three
reference distributions, on account of its concentration in a narrow
interval around $a/(N+1) = .1$. This {\it weak extreme} function is
included among the reference functions more for the formal reason of
its extremity rather than its applicability in any instance. The
discontinuous jump in the conditional probability
value of $P(E_{N+1}|S_N=a)$ from $p_L(N)$ to $\frac{a_1}{N}$ which it portrays as
$a/N$ crosses the threshold value of $\frac{a_1}{N}$ is not very
realistic.  It is interesting that to the contrary, the Linear function tracks closely 
with the Strict Extreme.  The Quartic function appears as
intermediate between the Weak and Strong Extreme functions. These
comments
are relevant both when $N = 100$ and when $N = 1000$.

The constriction of the reference distributions becomes even more
dramatic as N increases further. This can be seen in the lower graph
of Figure~\ref{fig:pdfpicturebig} which displays much more detail
for the size of $N = 10^5$, though the display is limited here to the Linear and Strict Extreme functions. For now, the
only new functions important to notice in the lower panel of 
Figure~\ref{fig:pdfpicturebig} are the virtually identical
solid-lined ``density'' functions labeled $L_1$ and $S_1$, pertinent
to $N_1 = 10^5$. For comparison purposes, the dash-dot-dash lined
functions $(^{\_\;.\;\_\;.\;\_})$ labeled $L_2$ and $S_2$ are
replicas of the pdf's for L and S when $N = 100$, exhibited here on
the same scale via the normalised histogram ``density'' transforms.
The remaining two (nearly identical) purely dash-lined functions $(^{\_\ \_\
\_})$ labeled $L_R$ and $S_R$ that also appear in
Figure~\ref{fig:pdfpicturebig} will be discussed separately in
Section \ref{sect:issue} on ``reduction probabilities''.

\begin{figure}[!ht]
\begin{center}
\includegraphics[width=1\linewidth]{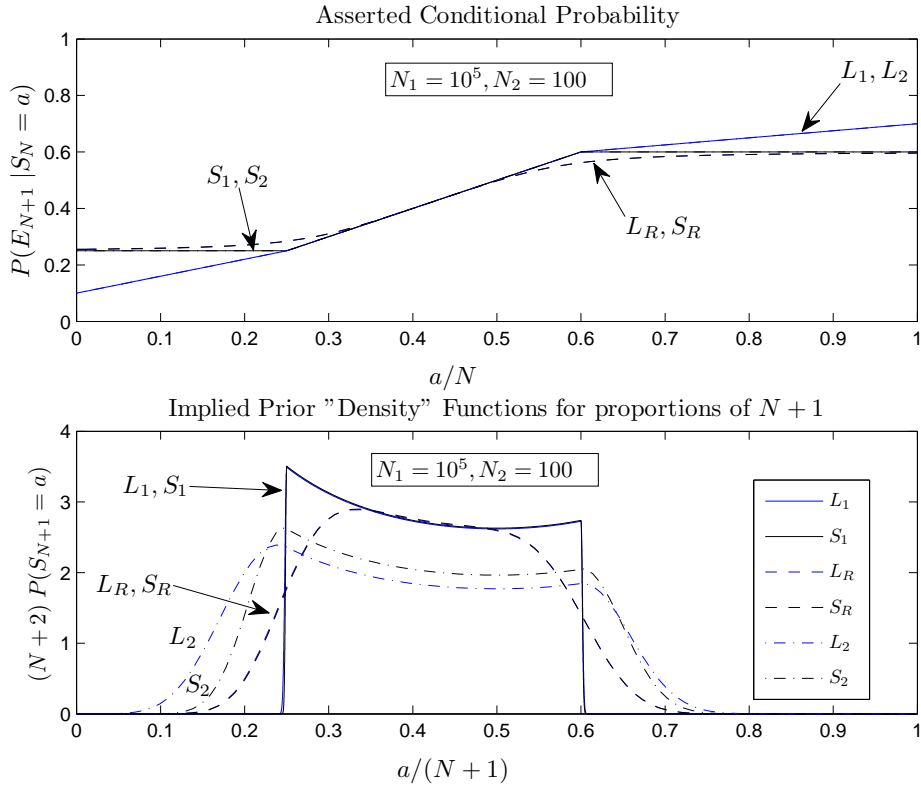}
\caption{The solid lines of the upper graph display conditional
probability functions relevant to the L and S specifications both
for $N_1 = 100,000$ and for $N_2 = 100$.  In the lower graph, the
associated $L_1$ and $S_1$ mass functions are virtually indistinguishable
solid lines when $N_1 = 100,000$.  The $L_2$ and $S_2$ mass functions for $N_2 = 100$
are distinguishable as the
pair of ``dash-dot-dash'' functions. Additionally, the lower graph
includes a pair of indistinguishable ``dash-dash-dash''
functions labeled $L_R$ and $S_R$. These two
functions, along with the associated function labeled $L_R, S_R$ in the
upper Figure are discussed in the ``Example of Complete Reductions'' in
Section \ref{sect:issue} They derive from a ``reduction'' of
the specification for $N_1 = 100,000$
to implied functions relevant to $N_R = 100$.}
\label{fig:pdfpicturebig}
\end{center}
\end{figure}

The $L_1$ and $S_1$ functions displayed in lower Figure \ref{fig:pdfpicturebig} exhibit clearly
a sharp bimodality in the {\bf q}$_{N+2}$ vector.  In this light it
can now be recognised that most all the ``density'' functions shown in
Figures 2 and 3 have this bimodal feature as well, though it is not so
readily apparent for the smaller values of $N$ seen in Figure \ref{fig:pdfpicture}. 

The computational software we have made available allows the
investigation of a number of sensitivity issues in the
specification of pmf's {\bf q}$_{N+2}$ via assertions in the form of
$PaN[a_1,a_2,p_L(N),p_U(N)]$.  A report on one such investigation in
Appendix 1 displays the sensitivity of the {\bf q}$_{102}$ vectors
shown in Figure 2 (Top) to the specification of $p_U(N)$, particularly
relevant to the case of the Weak extreme function.  More extensive comparisons 
here would detract from our focus in this introductory article on some pressing 
issues.

Our understanding of the {\bf q}$_{N+2}$ ``densities'' will deepen
as we now investigate the implications of the coherent reduction
equations (\ref{eq:usualreduction}) and (\ref{eq:reduction}) for the inference they imply for smaller sized
samples.

\section {Issues of reduction and extendibility}
\label{sect:issue}

Theorem \ref{thm:first} of Section \ref{sect:thm1},  which pertains to distributions that mimic
{\it all positive} conditioning frequencies, has some annoying
corollaries. 

For one, asserting $p_{a,N} = a/N$ for all the values
of $a = 1, 2, ..., N-1$ implies, via coherency, the assertion of
$p_{a,n} = a/n$ for $a = 1, 2, ..., n-1$ {\it for every integer} $n
< N$ as well ($n \geq 2$). %{\color{red} 
This can be seen by inserting $p_{a,N} =
a/N$ and $p_{a+1,N} = (a+1)/N$ on the right-hand-side of equation
(\ref{eq:reduction}) which yields $p_{a,N-1} = a/(N-1)$ for all
appropriate values of $a$. 

This is one of several coherence
properties that are problematic for frequentist estimates of
probabilities based on a ``large'' number of observations.
Interpreting conditional probabilities as ``estimates'' of ``the
probability'', coherency then would require frequentist estimates
{\it for any smaller number of observations as well}, no matter what
the observed relative frequency might be, as long as it does not equal $0$ or $1$.

In the context of frequency mimicking distributions over a limited
domain studied here, coherency also requires a specific reduction of
the conditional probability function {\bf p}$_{N+1}$ to lower orders
{\bf p}$_{n+1}$, but frequency mimicking is required only for a
limited range of values $n < N$ and for a limited range of values
for ``a'' among the component conditional probabilities $p_{a,n}$.
The implications are similar to those when the frequency mimicking
domain is unlimited, but they are not universal. We report them as
Theorem \ref{thm:third}   in Section  \ref{sect:reductive}  and Theorem \ref{thm:fourth} in Section  \ref{sect:acautionary} . They motivate the usefulness of
frequency mimicking assertions over a limited subdomain for
inference in applications to finite population problems.

\subsection{Reductive implications of restricted frequency mimicking}\label{sect:reductive}

Theorem \ref{thm:third} specifies how the assertion of frequency mimicking
conditional probabilities over a limited subdomain implies frequency
mimicking probabilities for a specifically limited number of shorter
event sequences, and over increasingly
more restricted subdomains.
\begin{theorem}\label{thm:third}
The frequency mimicking assertions of \
$PaN[a_1, a_2, p_L(N), p_U(N)]$ \ for $N+1$ events regarded
exchangeably imply via coherency the concomitant frequency mimicking assertions of \
$Pan[a_1, a_2 - (N-n), p_L(N), p_U(N)]$ \ for each smaller integer $n$
within $[n_0, N]$, \ where
$n_0 \ = \ N-(a_2-a_1)$. 	
\end{theorem}
%\noindent {\bf Theorem 3:} The frequency mimicking assertions of \
%$PaN[a_1, a_2, p_L(N), p_U(N)]$ \ for $N+1$ events regarded
%exchangeably imply via coherency the concomitant frequency mimicking assertions of \
%$Pan[a_1, a_2 - (N-n), p_L(N), p_U(N)]$ \ for each smaller integer $n$
%within $[n_0, N]$, \ where
%$n_0 \ = \ N-(a_2-a_1)$.  \ \ \ \ \ \ \ \ \ \ \ \ \ \ \ \ \ \ \ \ \ \  $\diamond$ 
\begin{proof} The assertion of
	$PaN[a_1,a_2,p_L(N),p_U(N)]$ amounts to the assertion of
	$(a_2-a_1+1)$ distinct frequency mimicking conditional
	probabilities, $p_{a,N} = a/N$ for integers $a \in [a_1, a_2]$. The
	application of the reduction equation (\ref{eq:reduction}) to each adjacent pair of
	these, $(p_{a,N}, p_{a+1,N})$, yields $(a_2-a_1)$ frequency
	mimicking probabilities at the level of $(N-1)$ conditioning events:
	$p_{a,N-1} = a/(N-1)$ for each integer $a \in [a_1,a_2-1]$.
	Repeating such reductions sequentially applied to these values of
	$p_{a,N-1}$ yields similar FM probabilities at the next lower level:
	$p_{a,N-2} = a/(N-2)$ for integers $a \in [a_1, a_2-2]$.  Continuing
	such  reductions iteratively $(a_2-a_1)$ times yields a final single
	mimicked frequency as
	$p_{a_1,n_0} = a_1/n_0$ at the smallest size
	of $n_0 = N - (a_2-a_1)$\vspace{.08cm}.\\
	\indent Furthermore, again applying the reduction equation (\ref{eq:reduction}) to the 
	monotonic non-decreasing values of adjacent pairs $(p_{a,N}, p_{a+1,N})$ 
	yields similarly monotonic 
	pairs of
	$p_{a,N-1}$ for the 
	next lower level of $N-1$\vspace{.06cm}. 
	The inequalities  $p_{a-1,N} \leq p_{a,N} \leq  p_{a+1,N}$  imply $p_{a-1,N} \leq p_{a,N-1}$ \vspace{.1cm} 
	because the latter is equivalent to $\frac{p_{a-1,N-1}}{1-p_{a,N}+p_{a-1,N}} \leq \frac{p_{a,N}}{1-p_{(a+1),N}+p_{a,N}}$ \vspace{.06cm} on the basis of (\ref{eq:reduction}).  
	This reduces to \vspace{.1cm} $p_{a-1,N}(1-p_{a+1,N}+p_{a-1,N}) \leq p_{a,N}(1-p_{a,N}+p_{a-1,N})$ and then by simple algebra to $p_{a-1,N}(1-p_{a+1,N}) \leq p_{a,N}(1-p_{a,N})$.
	\vspace{.06cm}
	Both of these paired factors are appropriately ordered, because 
	$p_{a-1,N} \leq p_{a,N}$ and $(1-p_{a+1,N}) \leq (1-p_{a,N})$ on account of the monotonicity inhering in the assertion of $PaN[a_1,a_2,p_L(N),p_U(N)]$ .\vspace{.06cm}\\
	\indent As to the lower and upper bounds for $p_{0,n}$ and $p_{n,n}$
	for these lower values of $n < N$, which remain specified as
	$p_L(N)$ and $p_U(N)$ in the Theorem, the only requirements for
	their coherency are that $p_{0,n} < p_{a_1,n}$ and $p_{n,n} >
	p_{a_2-(N-n), n}$.   In each such instance of n, the inequalities \
	$a_1/n > a_1/N > p_L(N)$ \ and \ $[a_2-(N-n)]/n < a_2/N < p_U(N)$
	are satisfied.  Thus, the specification of $p_L(N)$ and $p_U(N)$ as
	bounds for $p_{0,n}$ and $p_{n,n}$ are satisfactory.
	\ \ \ \ \ \ \ \ \ \ \ \ \ \ \ $\qed$ 
\end{proof}

Appendix 2 displays the structure of these implied reductions
geometrically in an insightful way.

 \noindent {\bf Numerical Example:}
On the basis of Theorem 3, $Pa100[25,60,.1,.7]$ implies the further
sequence of assertions $Pa99[25,59,.1,.7]$, $Pa98[25,58,.1,.7]$,
..., $Pa65[25,25,.1,.7]$. The range of frequencies that must be
mimicked by conditional probabilities diminishes as the number of
conditioning events, $n$, diminishes; and the lowest size of
conditioning observations that require frequency mimicking is
specifically limited to $n_0 \ = \ N-(a_2-a_1) = 100-(60-25) = 65$.
Notice that the size of $a_1$ remains fixed at $25$ throughout the
reduction process, while $a_2$ diminishes sequentially until it also
equals $a_1$. %{\color{red}
The only conditional probability that necessarily mimics a frequency based on
$S_{65}$ is $P(E_{66}|S_{65}=25) = 25/65$. 

 This numerical example
pertains specifically to the assertions we
have discussed for the apiary situation. %  $\ \ \diamond$ 

The assertions of $PaN[a_1,a_2,p_L(N), p_U(N)]$ place limits on the
extent to which conditional probabilities {\it must} mimic
conditioning frequencies. If these $PaN$ assertions were augmented 
by further assertions of $p_{a,N}$ values that do {\it
not} mimic frequencies {\it outside} of the interval
$[\frac{a_1}{N}, \frac{a_2}{N}]$ such as the L, Q, W, or S completions 
shown in Figure~\ref{fig:panpicture}, 
then reduction
equation (\ref{eq:usualreduction}) could be applied sequentially to
the implied probability mass function to determine pmf vectors
for the sums of smaller numbers of sample observations as well.  While these reduced 
distributions are not formally FMD's, they have great practical 
interest, as we shall now see. 

\noindent {\bf An Example of Complete Reduction:}
Refer once again to Figure~\ref{fig:pdfpicturebig} to study the 
following example. Suppose
that the FMD assertions $Pa100000[25000, 60000, .1, .7]$ are
augmented by either Linear or Strict completions.  The two implied pmf vectors {\bf q}$_{100002}$ are virtually
indistinguishable, looking like a box with a convex curved top,
labeled $L_1,S_1$ in the lower half of
Figure~\ref{fig:pdfpicturebig}.  The two indistinguishable {\it
purely dashed} functions there, labeled $L_R,S_R$, depict the
cohering pmf's {\bf q}$_{102}$ for the sum of only 101 events that
have been {\it reduced} from these vectors {\bf q}$_{100002}$ via equation
(\ref{eq:usualreduction}). These resulting ``density'' functions for
$S_{101}$ appear pleasingly regular. They are more concentrated than
the linear and strict FMD's $L_2$ and $S_2$ specified
directly via $Pa100[25,60,.1,.7]$.  These two pmf's are displayed on the same graph
as {\it dash-dot-dash} functions merely for comparison. 

\indent The pmf's $L_R$ and $S_R$ for this reduced distribution
display several interesting features. Although Theorem 3 assures
that frequency mimicking distributions are implied on sizes of $n$
only as low as $n_0 \equiv N -(a_2-a_1) = 100000-(60000-25000) =
65,000$, the reduced distribution for $S_{100}$ is also {\it very
nearly frequency mimicking} over most of the interval
$[\frac{a_1}{N}, \frac{a_2}{N}] = [.25, .60]$. Look at the purely
dashed conditional functions $L_R$ and $S_R$ in the upper panel of
Figure~\ref{fig:pdfpicturebig}. Frequency mimicking is almost exact
(to the resolution of the eye) over the interval $(.34, .51)$ and is
not far out of line anywhere over the interval $[.25, .60]$. Notice
also in the Figure that the lower and upper bounds on $p_{a,100}$
for this reduced conditional probability function have shifted to $.25$ and $.60$
\ (which equal the values of $\frac{a_1}{N}$ and $\frac{a_2}{N}$)   from the original assertion values of 
$p_L(100000) = .1$ and $p_U(100000) = .7$  respectively. Moreover,
the reduced functions $L_R$ and $S_R$ appear much less severe than
the two indistinguishable ``densities'' for $S_{100000}$, having
lost their bimodality. 

\indent All in all, this computation is pleasing.  The assertions $Pa10^5[25000,60000,.10,.70]$ determine a pmf {\bf q}$_{100002}$ that
has sensible implications both for inference on the basis of $100$
conditioning events and for opinions about the sum of $101$ events.
$\ \ \ \ \ \ \ \ \ \ \ \ \ \ \ \ \diamond$

\subsection{A cautionary result on extensions}\label{sect:acautionary}
If you are willing to assert $PaN[a_1, a_2, p_L(N), p_U(N)]$ for
some size of N, it would not seem surprising that you may like to
assert a frequency mimicking conditional probability for larger
sizes of N, too, especially when the conditioning frequency still
lies within the interval $[\frac{a_1}{N}, \frac{a_2}{N}]$. We shall
now see that this would surely be
coherent, and you may judge this to be appropriate. However, if you do wish to extend your FM assertions even the
smallest bit in this way for the ``next value of N'', coherency {\it
forces you to extend} your predictive probabilities as frequency
mimicking assertions {\it outside of the interval} $[a_1/N, a_2/N]$
as well. Theorem 4 makes this coherency condition explicit for
``the next value of $N$''.

%\noindent {\bf Theorem 4:} The assertions $PaN[a_1, a_2]$ along with
%the further assertion of $P(E_{N+2}|S_{N+1} = a) = a/(N+1)$ for any
%single integer value of $a$ within $[a_1, a_2]$ jointly imply all
%the assertions represented by $Pa(N+1)[a_1, a_2+1]$.  \ \ \ \ \ \ \
%\ \ \ \ \ \ \ \ \ \ \ \ \ \
%\ \ \ \ \ \ \ \ \ \ \ \ \ \ $\diamond$\\

\begin{theorem}\label{thm:fourth}
%\noindent {\bf Theorem 4:} 
Suppose that a further event, $E_{N+2}$ is appended to the vector of events ${\bf E}_{N+1}$ that are regarded exchangeably.  It is coherent to extend the assertions
entailed in $PaN[a_1, a_2, p_L(N), p_U(N)]$ to include the further
assertion of $P(E_{N+2}|S_{N+1} = a) = a/(N+1)$ for any specific integer value of $a$ for which $a/(N+1)$ is within the interval $[a_1/N, a_2/N]$.  However, coherency then also
requires the frequency mimicking assertions of $p_{a,(N+1)} = a/(N+1)$ {\it for every integer} value of $a$ within the interval
$[a_1, a_2+1]$.  
\end{theorem}

\noindent {\bf Comments:} \ Notice firstly that this implication extends the
FM interval based on (N+1) events to be wider than that based on N
events, because\\ %\begin{center}
\hspace*{2.7cm}$a_1/(N+1) \ < \ a_1/N \ < \ a_2/N \
< \ (a_2+1)/(N + 1)\;$.\\  
%\end{center} 
\noindent Thus, the subdomain of the
FMD's within $(0,1)$ is extended from the interval $[\frac{a_1}{N},
\frac{a_2}{N}]$ to
$[\frac{a_1}{N+1}, \frac{a_2+1}{N+1}]$. Secondly, if an extension of the lower bound $p_L(N+1)$ were entertained as well, it would need to be specified at a level not exceeding $a_1/(N+1)$.  A similar qualification would pertain to any assertion of $p_U(N+1)$.  Such bounding assertions would complete a full
assertion in the frequency mimicking form $Pa(N+1)[a_1, a_2+1,
p_L(N+1), p_U(N+1)]$.  These qualifications regarding further assertions of $p_L(N+1)$ and $p_U(N+1)$ derive from the presumptions that the assertion values $p_{a,(N+1)}$ are nondecreasing with the size of $a$. 
% Otherwise, the values of $p_L(N)$ and $p_U(N)$ in the statement of the theorem are irrelevant to the proof. \\
\begin{proof}
Equations (\ref{eq:usualreduction}) and (\ref{eq:reduction}) hold for any value of
$N$ for which $N+1$ events are regarded exchangeably.  If the assertion $PaN[a_1, a_2, p_L(N), p_U(N)]$ is extended so that $N+2$ events are regarded
exchangeably, and the FM assertion $p_{a,N+1} = a/(N+1)$ is added to
those of $PaN[a_1,a_2,p_L(N),p_U(N)]$, then equation (\ref{eq:reduction}) would
expand to the requirement of $(\ref{eq:reduction}')$ which would now be viewed as an
extension requirement:
\[
p_{a,N} \ \ = \ \ \frac{p_{a,N+1}}{
	1 - p_{a+1,N+1} + p_{a,N+1}} \ ,
\ \ \ \ {\rm for} \ \ a \ = \ 0, 1, ..., N\ .  \tag{\ref{eq:reduction}$'$}
\label{eq:reductiononehigher}
\]
\noindent Consider the value of $a \in [a_1, a_2]$ for which the
frequency mimicking extension is proposed. Inserting into equation
$(\ref{eq:reduction}')$ the values of $p_{a,N} = a/N$ and $p_{a,N+1} = a/(N+1)$
according to the conditions of Theorem 4 yields $p_{a+1,N+1} =
(a+1)/(N+1)$, a frequency mimicking assertion as well.  Continue
iteratively with this procedure for successive values of $a$ and $a+1$
until arriving at the implication $p_{a_2+1,N+1} =
(a_2+1)/(N+1)$.\vspace{.2cm}\\
\indent Similarly, equation $(\ref{eq:reduction}')$ can be written in a downward
direction with respect to ``$a$'':
%\begin{equation*}
%\noindent p_{a-1,N} \ \ = \ \ p_{a-1,N+1} \ / \ (1 - p_{a,N+1} + p_{a-1,N+1}) \ ,
%\ \ \ \ {\rm for} \ \ a \ = \ 1, 2, ..., N+1\ . \ \ \ \ \  (5'') \label{eq:reductiononelower}
%\end{equation*}
\[
p_{a-1,N} \ \ = \ \ \frac{p_{a-1,N+1}}{
	1 - p_{a,N+1} + p_{a-1,N+1}} \ ,
\ \ \ \ {\rm for} \ \ a \ = \ 1, 2, ..., N+1\ .  \tag{\ref{eq:reduction}$''$}
\label{eq:reductiononelower}
\]
Now insert the values of $p_{a-1,N} = (a-1)/N$ and $p_{a,N+1} =
a/(N+1)$ here to yield the  result $p_{a-1,N+1} =
(a-1)/(N+1)$. Repeat such insertions sequentially until arriving at the implication $p_{a_1,N+1} =
a_1/(N+1). \ \ \ \ \ \ \ \ \ \ \ \qed$
\end{proof}
\noindent  A geometrical exposition of the content of Theorem \ref{thm:fourth}
provides further insight in Appendix 3, continuing the analysis reported in
Appendix 2.

\subsection{Reasons for caution are corollary}\label{sect:5.3reasons}

It may seem appealing to augment a group of $PaN[a_1, a_2, p_L(N),
p_U(N)]$ assertions even more expansively.  Suppose you assert
$Pa100[25,60,.1,.7]$, i.e., frequency mimicking conditional
probabilities for $N = 100$, with $a_1 = 25$ and $a_2 = 60$, with
lower and upper bounds on $p_{0,100}$ and $p_{100,100}$ as $.1$ and
$.7$. Would you not then also want to assert similarly frequency
mimicking probabilities $P(E_{N+K+1}|S_{N+K} = a) = a/(N+K)$ for
{\it any} $K
> 0$ {\it and for every value of} $a$ {\it for which} $a/(N+K)$ lies
within the rational interval $[\frac{a_1}{N}, \frac{a_2}{N}]$?
Although such a general extension may seem reasonable, a Corollary
to Theorem \ref{thm:fourth} tells us that
coherency would force you into further frequency mimicking
assertions even
more extensive (over much wider intervals) than you might wish to bargain for, at least for large values of $(N+K)$. These exhibit themselves in their implications for the pmf vector {\bf q}$_{N+K+2}$.
\begin{corollary}\label{cor:first}
%\noindent{\bf Corollary 1:} 
%For any positive integer value of K, the 
Assertions of 
$PaN[a_1,a_2, p_L(N), p_U(N)]$ can be extended coherently to FMD's over $(N+K)$ events for any $K$ by augmenting them with 
assertions of $P(E_{N+j+1}|S_{N+j} = a) = a/(N+j) \in [a_1/N,
a_2/N]$ \; for any positive integer values of $j = 1, 2, ..., K$.
For any such $K$, coherency then requires 
frequency mimicking assertions $p_{a,N+K+1}$ over the wider interval $[a_1/(N+K), (a_2+K)/(N+K)]$.%$\ \ \ \diamond$
\end{corollary}
\begin{proof}
%\noindent{\bf Proof:} 
 This result derives simply from a finite
iterative application of Theorem \ref{thm:fourth}.  At each step of increasing
values of $K$, $a_1$ remains fixed at $a_1$, whereas the applicable
value of
``$a_2$'' increases by $1$, eventually to $a_2+K$. $\ \ \ \qed $ 
\end{proof}
\noindent {\bf Comments:}  Notice firstly that for every value of
$K$, both endpoints of the implied frequency mimicking interval, $a_1/(N+K)$ and $(a_2+K)/(N+K)$, lie outside of the asserted
frequency mimicking interval $[\frac{a_1}{N}, \frac{a_2}{N}]$.
Moreover, they continue to move further away from these bounds and
even approach the open unit-interval $(0,1)$ as K increases.
Secondly, in cases for
which the extension of FMD's might be desirable, it would be natural
to assert broader bounds $p_L(N+K) \leq a_1/(N+K)$ and $p_U(N+K) \geq (a_2+K)/(N+K)$ as well, appropriate to
the more extreme conditions to which they pertain, i.e., $(S_{N+K}=0)$ and
$(S_{N+K}=N+K)$, respectively. 

Understanding the full weight of the implications stated in
Corollary \ref{cor:first} comes from studying the limiting distribution of the
proportion of successes as the extension number $K$ increases. We
shall discuss this issue in the next subsection, providing an
insight into the nature of finitely additive distributions that
exhibit adherent masses. 

\subsection{The limit of distributions for the proportion}\label{sect:5.4limit}

Exchangeable distributions are most widely known
on account of de Finetti's representation theorem.  It says that if a
sequence of events $E_1, ..., E_{N+1}$ is regarded exchangeably and
as infinitely exchangeably extendible, then for any $a$ and $N$, \vspace{.2cm}\\
\noindent \hspace*{.2cm}
$P(S_{N+1} = a) \ = \ P(\overline{S}_{N+1} = \frac{a}{N+1}) \ =
\left({N+1\atop a}\right)\; \int_0^1 \theta^a\; (1-\theta)^{N+1-a}\;
dM(\theta)$\ \vspace{.2cm},\\
\noindent where \ $M(\theta) \; = \; \lim_{K\rightarrow\infty} M_{\overline{S}_{N+K}}(\theta)  $\ 
% \hspace{3.5cm}
for some sequence of finitely additive distributions
\{$M_{\overline{S}_{N+K}}(^.)$\}$_{K=1}^\infty$.\ \;
See \citet*{Heath76}, \citet*[pp. 87--89]{Landenna86} or \citet*[pp. 207--209]{Lad96}.
%  Heath and Sudderth (1976), Landenna and Marasini (1986, pp. 87-89) or Lad (1996, pp. 207-209).
\citet{Diaconis80}
% Diaconis and Freedman (1980) 
noted that even if the distribution is
exchangeably extendible only to N+K, then for some such distribution
$M(^.)$ this mixture representation differs from the actual value of
$P(S_{N+1} = a)$ by at most $4N/(N+K)$ for any $a$.  In
practice, the mixing distribution $M(^.)$ in the representation
theorem (or the ``prior distribution for $\theta$'' as it is commonly referred to) is
meant to represent one's initial opinions about the relative
frequency of success in an arbitrarily large sequence of events that
one would regard exchangeably with the
events composing {\bf E}$_{N+1}$. 

The important distributions in all real problems of
practice are the finite members of the sequence
\{$M_K({\bar{S}_{N+K}})$\}, {\it not} the limit of this sequence.  Nonetheless, in the 
context of extendible frequency mimicking distributions
over a limited domain, we can state precisely what happens to the
{\it limiting distribution of the frequencies} $\bar{S}_{N+K}$ as K increases. (Note again, this is something different from the distribution of the limit of the frequencies.)  Rather than continuing with formalities of theorem and proof, we shall discuss the development of the sequence informally here to develop intuitions.  Formalities and proofs are deferred to the Supplementary Materials available for this article.  These amount to a formalization of Sections \ref{sect:5.4limit}  and \ref{sect:5.5limitfmd} of this article.

Further to the conditions of Theorem \ref{thm:fourth} and Corollary \ref{cor:first}, which imply
the assertions of \ $Pa(N+K)[a_1, K+a_2, ^., ^.]$ for every $K
\geq 1$, \; any probability mass function denoted by {\bf q}$_{N+K+2}$
is restricted to have only $a_1 + (N - a_2)$ free components. As the
value of $K$ increases, the tendency in all agreeing distributions
is for virtually all the mass in the vector {\bf q}$_{N+K+2}$ to
settle essentially on two points, $S_{N+K+1} = a_1-1$ and $S_{N+K+1}
= a_2+K+1$.  When divided by $(N+K+1)$, the positions of these two points
of amassment converge toward $0$ and $1$ as
$K$ increases. This is the source of the ``adherent masses'' at $0$
and $1$ in the limiting distribution for the sequence of finitely
additive distributions.

Thus, the limit of this sequence of
distributions for $\bar{S}_{N+K}$ is an unusual one.  It is improper
and finitely additive, assigning probability $0$ to the points $0$
and $1$ {\it and} to every open interval {\it strictly} within $(0,1)$.  Yet the total mass of $1$ becomes stuck onto the endpoints of the unit interval. In this way it
exhibits what 
\citet{deFinetti1949,deFinetti55}
%de Finetti  (1949, 1955) 
called ``adherent masses'' at $0$ and $1$.  For the limiting distribution function allows the concomitant feature that  $P(0,a)+P(b,1) = 1$ for any $0 < a < b < 1$.  
(Recall our Definition  \ref{def:first} near the end of Section \ref{sect:intro} of this article.)  

Formal algebraic details and a statement of the relevant theorem can be read in the supplementary materials available in Appendix $4$ of this article.
We would best conclude this discussion with a computational graphic example.

\begin{figure}[!ht]
\begin{center}
\includegraphics[width=1\linewidth]{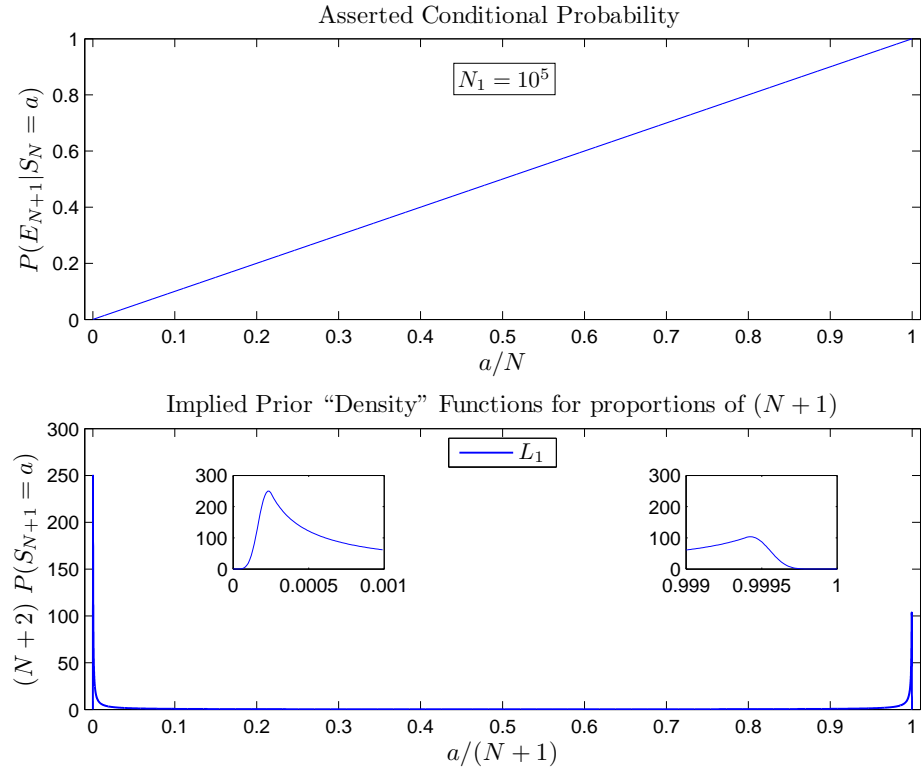}
\caption{The lower half of this Figure displays the
``density'' for $\bar{S}_{100001}$ implied by frequency
mimicking assertions for $N = 100$ and $N+K = 10^5$ along with the specifications
of $a_1 = 25, a_2+K = 99960, p_L(N+K) = .00012, p_U(N+K) = .99998$, and Linear completion probabilities, $L_1$.
This exemplifies the implications of
extending the frequency mimicking assumptions from $N = 100$ to $N+K = 10^5$ as long as the frequency is within the
interval $[.25, .6]$, which were formalised in Theorem 4.  The upper half of this
figure shows that frequency mimicking is now required over virtually
the entire unit-interval.  The lower half displays the peaks in the
density now amassed near to $\frac{24}{100001}$ and $\frac{99961}{100001}$. The relative heights of the peaks displayed in this Figure and the relative sizes of the adherent masses at $0$ and $1$ do depend on the completion formulation, presumed here as Linear.}
\label{fig:pdfpictureextrabig}
\end{center}
\end{figure}

\noindent {\bf A computational example:}  We can exemplify the
scenario developed in this discussion graphically. Suppose that frequency
mimicking is asserted initially for $N$ as low as 100, say
$Pa100(25,60,.1,.7)$ as in our apiary example, {\it and then
extended} to any larger $N+K$ as large as $10^5$ as long as
$a/(N+K)$ were within the interval $[.25, .60]$ along with specified
bounds $p_L(100001)=.00012$ and $p_U(100001) = .99998$. Theorem \ref{thm:fourth} and
Corollary \ref{cor:first} then imply frequency mimicking over a much wider
interval, $[.00025, .99940]$, which is almost over the entire unit
interval. Figure~\ref{fig:pdfpictureextrabig} exhibits the
implications for the distribution $M(\bar{S}_{100001})$.  The lower
graph of the density in Figure~\ref{fig:pdfpictureextrabig} shows how the  endpoint spikes of agglutinating mass develop when $N+K$
increases to $10^5$ from $100$.  For the scale of N and K in this
Figure, the peaks of the ``density'' for $\bar{S}_{N+K+1}$ occur at
24/100001 and at 99961/100001.  This feature can be appreciated in the algebraic detail presented in the Supplementary Materials to this article. Notice that the
displayed density no longer seems appropriate to intuitions about a
problem like the beehive problem which motivated the developments of
this article.  The beekeeper would probably {\it not} want to assert
such probabilities for the proportion of hives with queen cells in
such a large population of hives.  It appears much more appropriate
to assert merely $Pa100000[25000,60000,.1,.7]$ and to reduce this
distribution to the distribution implied for inference based on, say
$N = 100$ (as seen in Figure \ref{fig:pdfpicturebig}) than to assert $Pa100[25,60,.1,.7]$
and then to extend this assertion so to honour frequency mimicking for
$N+K$ as high as $10^5$.  Ultimately then, with such an attitude, frequency mimicking for $N$ as small as $100$ is only approximate.  However, it surely is visually apparent over a meaningful interval according to Berkson's ``interocular traumatic test'', the scale of the eyeball touted by Savage.  See Edwards, Lindman and Savage (1963, p. 217).

\subsection{The limit of FMD's over a constrained interval}\label{sect:5.5limitfmd}

Having reached this conclusion about the applicability of FMD's to finite population problems, it is intriguing to investigate the limiting distribution for the family of finite distributions agreeing with the assertions $PaN[a_1,a_2,p_L(N),p_U(N)]$ for any $N$.  Rather than specifying an FMD for some size of $N$ and then extending it, taking the consequences of broader and broader frequency mimicking intervals with the extension, suppose we fix the limit of a frequency mimicking interval and study the limit of discrete FMD's that honour this specific interval for growing sizes of $N$.

Specifically, consider the limit of distributions that respect frequency mimicking behaviour over the largest rational interval within a constant real interval $[\theta_1,\theta_2]$ as the size of N increases.  While an algebraic derivation is again left to the Supplementary Materials, Figure \ref{fig08} will assist one in intuiting the following result:  that at least in the cases of Strict or Linear augmentations of the FMP specifications outside a real interval $[\theta_1, \theta_2]$, the limit distribution of the family of FMD's is a 4-parameter Incomplete Beta mixture of Binomial distributions, with Incomplete Beta mixing parameters $(\theta_1, \theta_2, 0, 0)$.

To exemplify this result, Figure \ref{fig08} displays an array of ``densities'' for $\bar{S}_N$ deriving from assertions of $PaN[a_1(N,.2), a_2(N,.6), p_L(N)=.1, p_U(N)=.8]$ for the values of $N = 100$, $N = 1000$, and $N = 10^5$.  The designations of $.2$ and $.6$ in the specifications of $a_1(N,.2)$ and $a_2(N,.6)$ specify $a_1$ to be the smallest value of $a_1$ for which $a_1(N,.2)/N \geq .2$, while $a_2$ is the largest value of $a_2$ for which $a_2(N,.6)/N \leq .6$.  The Figure also exhibits the limit of such distributions as N increases, via the $Incomplete Beta (.2, .6, 0, 0)$ density.  

Recall that while a $Complete \ Beta(\alpha, \beta)$ density function for $\theta$, which is proportional to $\theta^{\alpha-1}(1-\theta)^{\beta-1}$ over $\theta \in [0,1]$, allows only parameters $\alpha > 0$ and $\beta > 0$ for proper integration, the four parameter Incomplete Beta  density over $\theta \in (\theta_1,\theta_2)$ strictly within $(0,1)$ integrates naturally when $\alpha = 0$ and $\beta = 0$.  In such a case the proportionality constant for the density equals $\{log[\theta_2/(1-\theta_2)] - log [\theta_1/(1-\theta_1)]\}^{-1}$.  Although this density is zero outside the interval $(\theta_1,\theta_2)$, when it mixes corresponding Binomial distributions as prescribed by exchangeability, the mixture allows positive probabilities for appropriate rational values of the average successes across the entire spectrum of rationals within $[0,1]$.  

\begin{figure}[tbph]
\centering
\includegraphics[width=0.9\linewidth]{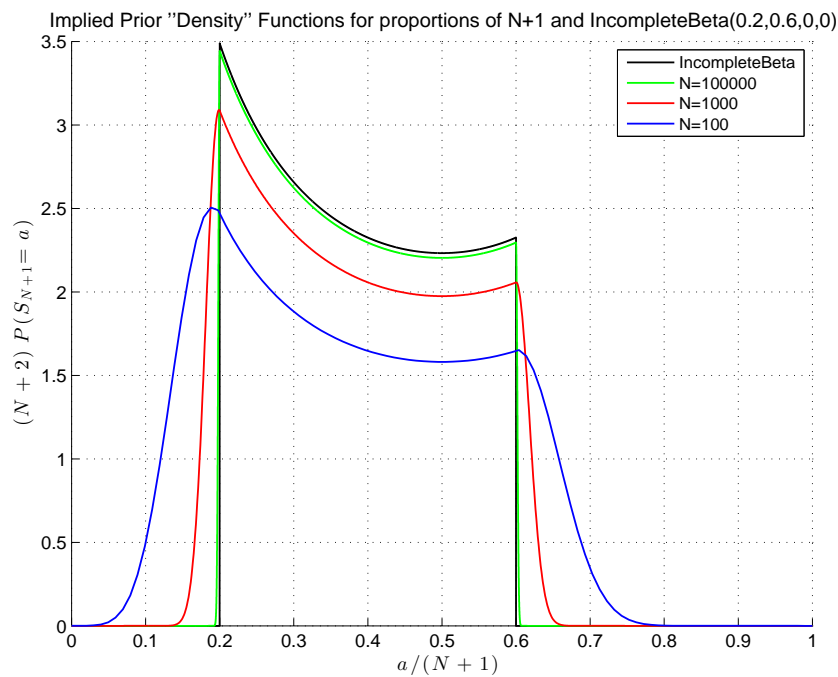}
\caption{Finite discrete frequency mimicking ``densities'' for a proportion, in the form of $PaN[a_1(N,\theta_1=.2), a_2(N,\theta_2=.6), p_L(N) = \theta_1/2 = .1, p_U(N) = 1-\theta_2/2 = .8$], \ for $N = 100, 1000,$ and $10^5$, along with the limiting density $Incomplete Beta (.2, .6, 0, 0)$ }
\label{fig08}
\end{figure}

\section{Concluding Comments}\label{sect:conclusions}

Emerging from the trees to view the forest, we can summarise the conclusions of this analysis.  In the context of FMD's specified by  $PaN[a_1,a_2,p_L(N),p_U(N)]$ predictive probabilities have been shown to be eminently applicable to inference from sampling in finite population problems when  $N+1$ is the total population size or smaller.  Such assertions do specify precise frequency mimicking assertions for a range of smaller values of N too, as well as virtual FMD's for even smaller values of N outside this range.  Moreover, the implications of such assertions for the distributions of observed frequencies of any number of observations have been specified.  However, the infinite extension of FMD assertions within the same interval as $[a_1/N, a_2/N]$ provokes typically unappealing conclusions.  Nonetheless, these are mathematically interesting for exhibiting a procedure for constructing finitely additive distributions that exhibit agglutinated masses, long recognised as an intriguing subject.  Finally, limiting distributions of the family $PaN[a_1(\theta_1,N),a_1(\theta_1,N),p_L(\theta_1,N),p_U(\theta_2,N)]$  have been derived which may be applicable to infinitely exchangeably extendible sequences. These limits are identified as Incomplete Beta mixtures of Binomial probabilities for the ``Linear'' and ``Strict'' subfamilies of FMD's. 

Throughout this article we have focused on the implications of FMD distributions for predictive probabilities for ``the next event to be observed'', identifying meaningful results for any size of $N$.  The structure
of the analysis can be applied  to other sensible quantities as well.  For example, the same computational strategy can produce inferential probabilities for the sum of the population characteristics conditioned on the sum of the sample characteristics, in the form of $P[(S_{N+1}=A)|(S_n=a)]$.
% I WAS GOING TO CLAIM THAT THE ONLY FEATURE OF EXCHBLE DISTNS WE USED WAS THE SEQUENTIAL 
% FORECASTING EQNS, AND THAT THUS THE AWE DISTNS OF LAD AND SCOZZAFAVA SUPPORT ALL OF THIS.
% BUT THIS IS NOT TRUE.  WE ALSO USE THE REDUCTION EQUATIONS FOR EVERY N+1 AND N, AND THUS
% FULL EXCHANGEABILITY IS REQUIRED.
%Finally, we should remark on the limitation of the analysis reported here to events that are %regarded exchangeably.  In fact, the only characteristics of exchangeable distributions that we %employed have been the  

\begin{acknowledgements}
%If you'd like to thank anyone, place your comments here
%and remove the percent signs.
We have been supported, with
thanks, by a travel grant from University of Rome ``La Sapienza'',
by a Ph.D. grant from University of Naples ``Federico II'', and by a
travel grant from the University of Palermo.  Thanks to James O'Malley, Wes
Johnson, Val Johnson, and Andrea Piesse for helpful comments on 
earlier drafts of this article during the many years of its development, 
and to Jay Kadane for a helpful reference.
\end{acknowledgements}

% PIPPO Bibliography, if accepted we can use this style for the biblio
%\cite{Berliner88}\cite{Berti09}\cite{Biazzo00}
%\cite{Bingham10}
%\cite{Brooke}
%\cite{Capotorti07}\cite{Cifarelli16}\cite{Cifarelli96}\cite{Coletti96}\cite{Diaconis80}\cite{Edwards63}\cite{deFinetti37}\cite{deFinetti1949}\cite{deFinetti55}\cite{deFinetti67}\cite{deFinetti72}\cite{Fortini12}\cite{Gilio16}\cite{Heath76}\cite{Hill88}\cite{Jeffreys39}\cite{Johnson05}\cite{Lad96}\cite{Lad96b}\cite{Kadane86}\cite{Lad93}\cite{Lad95}\cite{Lad92}\cite{Lad90}\cite{Landenna86}\cite{Martinez02}\cite{Regazzini87}\cite{Savage54}\cite{Savage61}\cite{Hill89}
% BibTeX users please use one of
\bibliographystyle{spbasic}        % basic
%\bibliographystyle{plainnat}      % basic
% style, author-year citations
\bibliography{bibliomimic}   % name your BibTeX data base

\begin{thebibliography}{36}
\providecommand{\natexlab}[1]{#1}
\providecommand{\url}[1]{{#1}}
\providecommand{\urlprefix}{URL }
\expandafter\ifx\csname urlstyle\endcsname\relax
  \providecommand{\doi}[1]{DOI~\discretionary{}{}{}#1}\else
  \providecommand{\doi}{DOI~\discretionary{}{}{}\begingroup
  \urlstyle{rm}\Url}\fi
\providecommand{\eprint}[2][]{\url{#2}}

\bibitem[{Berliner and Hill(1988)}]{Berliner88}
Berliner LM, Hill BM (1988) Bayesian nonparametric survival analysis. Journal
  of the American Statistical Association 83(403):772--779

\bibitem[{Berti et~al.(2009)Berti, Crimaldi, Pratelli, and Rigo}]{Berti09}
Berti P, Crimaldi I, Pratelli L, Rigo P (2009) Rate of convergence of
  predictive distributions for dependent data. Bernoulli 15(4):1351--1367

\bibitem[{Biazzo and Gilio(2000)}]{Biazzo00}
Biazzo V, Gilio A (2000) {A generalization of the fundamental theorem of de
  Finetti for imprecise conditional probability assessments}. International
  Journal of Approximate Reasoning 24(2-3):251--272

\bibitem[{Bingham(2010)}]{Bingham10}
Bingham NH (2010) Finite additivity versus countable additivity. Electronic
  Journal for History of Probability and Statistics 6:1--35

\bibitem[{Brooke et~al.(2003)Brooke, Kendrick, Meeraus, and Raman}]{Brooke}
Brooke A, Kendrick D, Meeraus A, Raman R (2003) Gams: a user's guide.
  Washington, D.C., GAMS Development Corp.

\bibitem[{Capotorti et~al.(2003)Capotorti, Galli, and Vantaggi}]{Capotorti03}
Capotorti A, Galli L, Vantaggi B (2003) Locally strong coherence and inference
  with lower--upper probabilities. Soft Computing 7(5):280--287

\bibitem[{Capotorti et~al.(2007)Capotorti, Lad, and Sanfilippo}]{Capotorti07}
Capotorti A, Lad F, Sanfilippo G (2007) {Reassessing Accuracy Rates of Median
  Decisions}. The American Statistician 61(2):132--138

\bibitem[{Cifarelli and Regazzini(1996)}]{Cifarelli96}
Cifarelli DM, Regazzini E (1996) De finetti's contribution to probability and
  statistics. Statistical Science 11(4):253--282

\bibitem[{Cifarelli et~al.(2016)Cifarelli, Dolera, and Regazzini}]{Cifarelli16}
Cifarelli DM, Dolera E, Regazzini E (2016) Frequentistic approximations to
  bayesian prevision of exchangeable random elements. International Journal of
  Approximate Reasoning 78:138 -- 152

\bibitem[{Coletti and Scozzafava(1996)}]{Coletti96}
Coletti G, Scozzafava R (1996) Characterization of coherent conditional
  probabilities as a tool for their assessment and extension. International
  Journal of Uncertainty, Fuzziness and Knowledge-Based Systems 04(02):103--127

\bibitem[{Diaconis and Freedman(1980)}]{Diaconis80}
Diaconis P, Freedman D (1980) Finite exchangeable sequences. The Annals of
  Probability 8(4):745--764

\bibitem[{de~Finetti(1949)}]{deFinetti1949}
de~Finetti B (1949) Sull'impostazione assiomatica del calcolo delle
  probabilit\`{a}. Annali Triestini dell' Universit\`{a} di Trieste 19:29--81,
  {M}aione G. (tr) On the Axiomatization of Probability Theory, in Probability,
  Induction, Statistics: the art of guessing (Wiley, London, 1972) Chapter 5,
  115--128

\bibitem[{de~Finetti(1952)}]{deFinetti1952}
de~Finetti B (1952) Gli eventi equivalenti e il caso degenere. Giornale
  dell'Istituto italiano degli Attuari 15:40--64, (tr) Equivalent events and
  the degenerate case, in Probabilit\`{a} e Induzione (Clueb, Bologna, 1993)
  129--152

\bibitem[{de~Finetti(1955)}]{deFinetti55}
de~Finetti B (1955) La struttura delle distribuzioni in uno spazio astratto
  qualsiasi. Giornale dell'Istituto Italiano degli Attuari 18, {M}aione G. (tr)
  The Structure of Distributions on Abstract Spaces, in Probability, Induction,
  Statistics: the art of guessing (Wiley, London, 1972) Chapter 7, 129--140

\bibitem[{de~Finetti(1967)}]{deFinetti67}
de~Finetti B (1967) Quelques conventions qui semblent utiles. Revue Roumaine
  des Mathematiques Pures et Appliquees 12(9):1227--1233, {S}avage L.J. (tr) A
  useful notation, in Probability, Induction, Statistics (Wiley, London, 1972)
  Appendix to introduction, xviii--xxiv

\bibitem[{de~Finetti(1970)}]{deFi70}
de~Finetti B (1970) Teoria delle probabilit\`{a}. Ed. Einaudi, 2 voll., Torino,
  english version: Theory of Probability 1 (2), Chichester, Wiley, 1974 (1975)

\bibitem[{Fortini and Petrone(2012)}]{Fortini12}
Fortini S, Petrone S (2012) Predictive construction of priors in bayesian
  nonparametrics. Brazilian Journal of Probability and Statistics
  26(4):423--449

\bibitem[{Fortini and Petrone(2016)}]{Fortini16}
Fortini S, Petrone S (2016) Predictive Distribution (de Finetti's View),
  American Cancer Society, pp 1--9. \doi{10.1002/9781118445112.stat07831},
  \urlprefix\url{https://onlinelibrary.wiley.com/doi/abs/10.1002/9781118445112.stat07831}

\bibitem[{Gilio et~al.(2016)Gilio, Pfeifer, and Sanfilippo}]{Gilio16}
Gilio A, Pfeifer N, Sanfilippo G (2016) Transitivity in coherence-based
  probability logic. Journal of Applied Logic 14:46--64

\bibitem[{Heath and Sudderth(1976)}]{Heath76}
Heath D, Sudderth W (1976) De {F}inetti's theorem on exchangeable variables.
  The American Statistician 30(4):188--189

\bibitem[{Hill(1988)}]{Hill88}
Hill BM (1988) De {F}inetti’s theorem, induction, and ${A}_n$, or {B}ayesian
  nonparametric predictive inference. In: Bernardo JM, Degroot MH, Lindley DV,
  M SAF (eds) Bayesian Statistics 3, Oxford University Press, pp 211--241

\bibitem[{Hill(1989)}]{Hill89}
Hill BM (1989) Bayesian nonparametric prediction and statistical inference.
  Fort Belvoir: Defense Technical Information Center, 29 pp.

\bibitem[{{Johnson} et~al.(2005){Johnson}, {Moosman}, and {Cotter}}]{Johnson05}
{Johnson} VE, {Moosman} A, {Cotter} P (2005) A hierarchical model for
  estimating the early reliability of complex systems. IEEE Transactions on
  Reliability 54(2):224--231

\bibitem[{Kadane et~al.(1988)Kadane, Schervish, and Seidenfeld}]{Kadane86}
Kadane J, Schervish M, Seidenfeld T (1988) Statistical implications of finitely
  additive probability. In: Zellner A, Goel P (eds) Bayesian Inference and
  Decision Techniques, Elsevier, pp 211--231

\bibitem[{Lad(1996{\natexlab{a}})}]{Lad96}
Lad F (1996{\natexlab{a}}) Operational subjective statistical methods: A
  mathematical, philosophical, and historical introduction. Wiley, New York

\bibitem[{Lad(1996{\natexlab{b}})}]{Lad96b}
Lad F (1996{\natexlab{b}}) Three useful applications of conditional
  probability, understood as an assertion of uncertain knowledge. Proceedings
  of the American Statistical Association, Section on Bayesian Statistical
  Science pp 47--52

\bibitem[{Lad et~al.(1990)Lad, Dickey, and Rahman}]{Lad90}
Lad F, Dickey J, Rahman M (1990) The fundamental theorem of prevision.
  Statistica 50(1):19--38

\bibitem[{Lad et~al.(1992)Lad, Dickey, and Rahman}]{Lad92}
Lad F, Dickey JM, Rahman MA (1992) Numerical application of the fundamental
  theorem of prevision. Journal of Statistical Computation and Simulation
  40(3-4):135--151

\bibitem[{Lad et~al.(1993)Lad, Deely, and Piesse}]{Lad93}
Lad F, Deely J, Piesse A (1993) Using the fundamental theorem of prevision to
  identify coherency conditions for finite exchangeable inference. Tech.
  Rep.~95, University of Canterbury Department of Mathematics and Statistics
  Research Report,
  \urlprefix\url{http://www.math.canterbury.ac.nz/php/research/reports/}

\bibitem[{Lad et~al.(1995)Lad, Deely, and Piesse}]{Lad95}
Lad F, Deely J, Piesse A (1995) Coherency conditions for finite exchangeable
  inference. Journal of the Italian Statistical Society 4:195--213

\bibitem[{Landenna and Marasini(1986)}]{Landenna86}
Landenna F, Marasini D (1986) Uno sguardo alle principali concezioni
  probabilistiche. Giuffr\`e, Milano

\bibitem[{Martinez and Martinez(2002)}]{Martinez02}
Martinez W, Martinez A (2002) Computational Statistics with Matlab. Chapman and
  Hall

\bibitem[{Regazzini(1987)}]{Regazzini87}
Regazzini E (1987) de finetti's coherence and statistical inference. The Annals
  of Statistics 15(2):845--864

\bibitem[{Savage(1954)}]{Savage54}
Savage LJ (1954) The Foundations of Statistics. Wiley, New York

\bibitem[{Savage(1961)}]{Savage61}
Savage LJ (1961) The subjective basis of statistical practice, unpublished
  University \indent of Michigan manuscript.

\bibitem[{Savage(1972)}]{Savage72}
Savage LJ (1972) The Foundations of Statistics, Second edition. Dover Books,
  New York

\end{thebibliography}

\newpage
\appendix
\section*{Appendix 1. Sensitivity of Weak extreme
distributions to $p_L$ and $p_U$}

\begin{figure}[!ht]
\begin{center}
\includegraphics[width=1\linewidth]{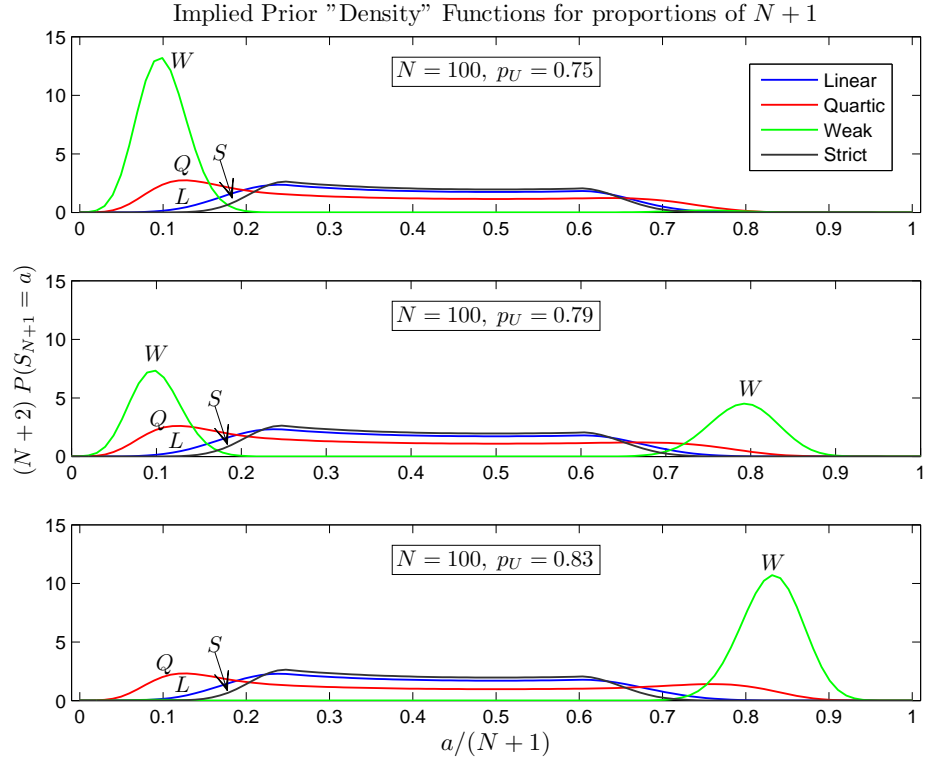}
\caption{The three panels of graphs all display mass functions for
$S_{102}$ cohering with the specifications of $N = 100$ along with
$p_L(100) = 0.10,\ \frac{a_1}{N} = 0.25$, and $\frac{a_2}{N} = .60$.
Only the specification of $p_U(100)$ varies between the displays:\ \
for the top display, $p_U(100) = .75$;\ \ for the middle display,
$p_U(100) = .79$;\ \ while for the bottom display, $p_U(100) =
.83$.}
\label{fig:pdfpicturenonrobust}
\end{center}
\end{figure}
 \indent Figure \ref{fig:pdfpicturenonrobust} displays an unusual
sensitivity of the Weak Extreme mass function vector {\bf q}$_{102}$ to the specification of
$p_U(N)$ in the assertions of $Pa100[a_1,a_2,p_L(100),p_U(100)]$.  The Figure shows three panels of mass functions specified exactly as in the top half of our maintext 
Figure~\ref{fig:pdfpicture} where $N=100$, but with the value 
of $p_U(100)$ changed sequentially from .70 in Figure ~\ref{fig:pdfpicture} to .75, .79
and .83 in the panels of Figure~\ref{fig:pdfpicturenonrobust}. With
these very mild changes in specification for the upper bound on
$P(E_{N+1}|S_N=100)$, the mass function for the
Weak extreme specification shifts dramatically across the spectrum
of the abscissa, shifting the location
of its main support from the interval $(.04, .18)$  to the interval
$(.75, .95)$.  Comparatively, this change in specification of
$p_U(100)$ has little effect on the distributions associated with
the Q, L and S functions, though Q has fattened to a noticeable
extent. Further
experimentation shows that it is the {\it relative
sizes} of $p_L(N)$ and $p_U(N)$ that drives the sensitivity.
Remember that the ``Weak'' specification of conditional
probabilities is an extreme distribution among FMD distributions.
\section*{Appendix 2.  A geometrical exposition of Theorem \ref{thm:third}}

\noindent This appendix presents a geometrical exposition of the
proof of Theorem \ref{thm:third} as it applies in the context of a single special
case. The algebraic content of the general proof in Section \ref{sect:reductive}
should then become intuitive. Understanding the detail will require
some serious attention, but we believe it will be worth it.  Let us
begin by restating Theorem \ref{thm:third} as it pertains to a special
case we shall illustrate, of $N = 8$: 

\noindent When $9$ events are regarded exchangeably, the assertion
of $Pa8[2,5, p_L(8),p_U(8)]$ \ implies via coherency the concomitant
assertions of 
 $Pa7[2,4,p_L(8),p_U(8)]$, $Pa6[2,3,p_L(8),p_U(8)]$,
and $Pa5[2,2,p_L(8),p_U(8)]$. 

The context of the following geometrical exposition has been introduced
previously in an article by Lad, Deely and Piesse (1995, pp.
200-201) which we review briefly now. Suppose that a conditional
probability $p_{a,N}$ is represented algebraically by a parametric
convention
\begin{equation}
p_{a,N} \ = \ \frac{a + \alpha_{a,N}}{N + \alpha_{a,N} +
\beta_{a,N}} \ \ \ , \label{eq:parametricpaN}
\end{equation}
for some pairs of positive numbers $(\alpha_{a,N}, \beta_{a,N})$.
This is an obvious generalisation of the parametric representation
of conditional probability from a Beta-Binomial mixture, or Polya
distribution. For that special case, the values of $\alpha_{a,N}$
and $\beta_{a,N}$ are fixed constants for all $a$ and $N$. Algebraic
transformation of equation (\ref{eq:parametricpaN}) shows that when
this more general equation holds, the conditional probability value $p_{a,N}$ can
be represented {\it by a specific line of such pairs}
$(\alpha_{a,N}, \beta_{a,N})$ via the expression \begin{equation}
\beta_{a,N} \ = \ -(N-a) \ + \ (a +
\alpha_{a,N})\;(\frac{1-p_{a,N}}{p_{a,N}}) \ \ \ .
\label{eq:parametricpaNline}
\end{equation}
\noindent Examples of such lines appear in Figure~\ref{fig:reductionfig} which we now discuss.

\begin{figure}[!ht]
\begin{center}
%\scalebox{1}{\input{fig05.pstex_t} }
\includegraphics{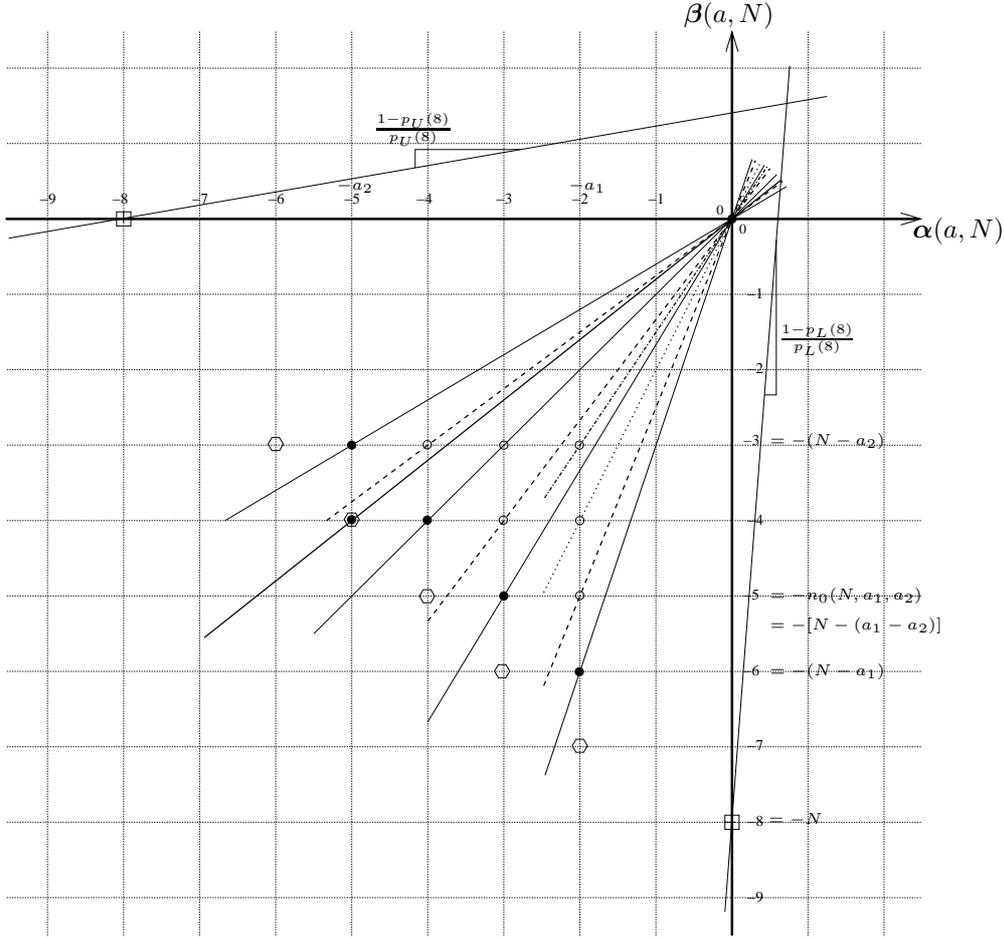}
\caption{Lines passing through the four {\it small dark filled} circles
represent the frequency mimicking conditional probabilities presumed
in Theorem 3 for the specific setup $Pa8[2,5,p_L(8),p_U(8)]$. Dashed and dotted lines passing
through small {\it open} circles represent further frequency mimicking
conditional probabilities implied by coherency.  The lines that pass through
the boxed points (0,-8) and (-8,0) represent the asserted lower limit $p_L(8)$ on
$p_{0,8}$ and the upper limit $p_U(8)$ on  $p_{8,8}$, respectively.  The points
designated by hexagons, along with the bold line that passes through
one of them, pertain to the suggestive graphical proof of Theorem \ref{thm:fourth}, to be 
discussed in Appendix 3.}

\label{fig:reductionfig}
\end{center}
\end{figure}

Equation (\ref{eq:parametricpaNline}) specifies a line passing through the point
$(-a, -(N-a))$ with a slope that equals the negative odds ratio
$(1-p_{a,N})/p_{a,N}$.   For example, the conditional probability
value $p_{4,8} = 4/8$ would be represented by a line through the
point $(-4,-4)$ with a slope
equal to $1$, as can be seen among the several lines
shown in Figure~\ref{fig:reductionfig}.  Specifically, \\
$\beta_{4,8} = -(8-4) + (4 + \alpha_{4,8}) (\frac{1-4/8}{4/8}) \ = \
\alpha_{4,8}\; $. \ \ This line, of course, also runs through the
origin $(0,0)$.  In fact, any line representing a frequency mimicking probability
$p_{a,N} = a/N$ must run through the origin and the
point $(-a, - (N-a))$, satisfying equation (\ref{eq:parametricpaNline}). 

Generally, the assertion of {\it any} numerical values whatsoever
for the conditional probabilities $p_{a,N}$ would be coherent, just
so long as the vector ${\bf p}_{N+1}$ lies within the unit-cube.
Each component $p_{a,N}$ would be represented by a line through $(-a, -(N-a))$
with a slope equal to the negative odds ratio that $p_{a,N}$
specifies. However, this line would not necessarily pass through the
origin. Only the lines representing frequency mimicking
probabilities
must do so.

Although coherency allows great freedom in the specification of
${\bf p}_{N+1}$, there is one important and indicative coherency
condition on lines representing components of the vectors ${\bf
p}_{N+1}$ and ${\bf p}_{N}$.  Recall the reduction equation
(\ref{eq:reduction}) discussed in Section \ref{sect:theinversion} which relate the three
conditional probabilities $p_{a,N-1}, p_{a,N}$, and $p_{a+1,N}$:
\begin{center}
$p_{a,N-1} \ \ = \ \ p_{a,N} \ / \ (1 - p_{a+1,N} + p_{a,N}) \ , \ \
\ \ {\rm for} \ \ a \ = \ 0, 1, ..., N-1\ $. \ \
(\ref{eq:reduction})
\end{center}
\noindent In the context of this geometrical representation of conditional probabilities by 
lines, two of these lines (representing $p_{a,N}$ and $p_{a+1,N}$)
pass through the diagonally adjacent points $(-a,-(N-a))$ and
$(-(a+1),-(N-(a+1)))$, respectively, with slopes appropriately equal
to the negodds ratios they specify. If they are both frequency mimicking
probabilities, these two lines must intersect at the origin $(0,0)$.
More generally, whatever their numerical values and wherever they
intersect, the geometrical implication of equation
(\ref{eq:reduction}) is that the third line representing
$p_{a,N-1}$, which passes through the point $(-a,-(N-1-a))$ must
also pass through the intersection point of the two lines
representing $p_{a,N}$ and $p_{a+1,N}$.  In the case of frequency
mimicking probabilities, this is again the origin.  The point
$(-a,-(N-1-a))$ which locates the line representing $p_{a,N-1}$ is
the third vertex of a right triangle whose other vertices are the
diagonally adjacent points $(-a,-(N-a))$ and $(-(a+1),-(N-(a+1)))$.
This structure of three intersecting lines can be
observed in several instances in Figure~\ref{fig:reductionfig}.

 Figure~\ref{fig:reductionfig} displays a
geometrical example of Theorem 3, specifically as it pertains to the
assertions we would denote by $PaN[a_1,a_2,p_L(N),p_U(N)] \ = \
Pa8[2,5,p_L(8),p_U(8)]$. The four lines that join each of the points
$(-2,-6), (-3,-5), (-4,-4)$ and $(-5,-3)$ with $(0,0)$ represent the
four frequency mimicking assertions constituting $Pa8[2,5,^.,^.]$, viz.,
$p_{2,8} = 2/8, \ p_{3,8} = 3/8, \ p_{4,8} = 4/8$ and $p_{5,8} =
5/8$. These lines through the origin are identified in the Figure by
{\it small darkened circles} on the four mentioned points. The slope
of each line through $(-a, -(N-a))$ and $(0,0)$ would equal the
relevant value of $(1-p_{a,N})/p_{a,N} \ = \ (N-a)/a$, which is the
conditional negodds ratio associated with a frequency mimicking
assertion for $p_{a,N}$.  The two lines through the boxed points
$(0,-8)$ and $(-8,0)$ whose slopes are labeled $(1-p_L(8))/p_L(8)$
and $(1-p_U(8))/p_U(8)$ respectively, represent the bounds asserted
for the probabilities $p_{0,8}$ and $p_{8,8}$.  These lines are
sloped more and less steeply than the lines representing $p_{2,8} =
2/8$ and
$p_{5,8} = 5/8$, respectively, a feature required by the fact that 
$PaN(^.,^.,^.,^.)$ assertions entail that values of $p_{a,N}$ are 
increasing in $a$.

\indent For the coherency reasons we have discussed above, the three
lines in Figure~\ref{fig:reductionfig} connecting the points $(-2,
-5), (-3,-4)$ and $(-4, -3)$ with $(0,0)$ then represent conditional
probabilities that are required to be frequency mimicking as 
well: \ $p_{2,7} = 2/7, p_{3,7} = 3/7$
and $p_{4,7} = 4/7$. For each of them must intersect the
intersection point of the lines running through points one unit to
their left and one unit below them.  These three lines are those
designated in the Figure by {\it small open circles} at the points
$(-2, -5), (-3,-4)$ and $(-4,-3,)$. Sequentially, the intersection
of any two adjacent lines from among these three must also intersect
at the origin along with a comparable line, one further size of $N$
down. Thus, Figure~\ref{fig:reductionfig} also includes these two
lines passing through the points $(-2,-4)$ and $(-3,-3)$,
representing $p_{2,6} = 2/6$ and $p_{3,6} = 3/6$, respectively.
Finally, the intersection of these two lines at $(0,0)$ requires one
final line, one further step down, through the points $(-2,-3)$ and
$(0,0)$. This represents the implied conditional probability
$p_{2,5} = 2/5$. As specified in the statement of Theorem \ref{thm:third}, the
value of $n_0 \ = \ N - (a_2 - a_1) \ = \ 8 - (5-2) \ = \ 5.$  All
lines implied by $Pa8[2,5,p_L(8),p_U(8)]$ appear as dashed lines
through open circles at appropriate points. These ten mentioned
lines exhaust the coherency conditions associated with the
assertions denoted by $Pa8[2,5,p_L(8),p_U(8)]$. These complete
implications of the Theorem can be enumerated:
$Pa8[2,5,p_L(8),p_U(8)]$ implies $Pa7[2,4,p_L(8),p_U(8)],
Pa6[2,3,p_L(8),p_U(8)],$ and $Pa5[2,2,p_L(8),p_U(8)]$. $\ \ \ \ \ \
\ \ \ \ \ \
\ \qed$

\indent The five additional points in Figure~\ref{fig:reductionfig}
designated by surrounding hexagons pertain to a continuation of this
example in Appendix 3.

\section*{Appendix 3.  A geometrical exposition of Theorem \ref{thm:fourth}}

Examine again the geometrical example in
Figure~\ref{fig:reductionfig}.  We have already discussed the
exhaustive extent of coherent implications of asserting
$Pa8[2,5,p_L(8),p_U(8)]$ for smaller values of $N$, as small as $n_0
= 5$. Now suppose you assert additionally one further conditional
probability, say $P(E_{10}|S_9 = 5) = 5/9$. (Notice that 5/9 lies
between 2/8 and 5/8, and you have already asserted frequency
mimicking conditional probabilities within this interval based on
eight conditioning events.) This additional assertion can be
represented by drawing another line into
Figure~\ref{fig:reductionfig} through the points $(-5, -4)$ and
$(0,0)$, which identifies $p_{5,9}\ = \ 5/9$. The point $(-5, -4)$
is identified in the Figure by both a dark-filled circle and a
surrounding hexagon. Based on your awareness of the geometrical
implication of the ``reduction'' equation (\ref{eq:reduction}), you
will now know from equation $(5^{\prime})$ that the line already
drawn through $(-5, -3)$ and $(0,0)$ must intersect this new line
through $(-5,-4)$ and $(0,0)$) at the same point as it intersects
still another line to be drawn through $(-6,-3)$, because this point
is diagonally adjacent to $(-5,-4)$. The point $(-6,-3)$ has been
surrounded by a small hexagon in Figure ~\ref{fig:reductionfig},
allowing you to draw a line through it and $(0,0)$ for yourself.
This line, required by coherency, means that you must now in
addition be asserting $p_{6,9} = 6/9$. For the same reason then,
moving down that diagonal array of points surrounded by hexagons,
the points $(-4,-5), (-3,-6)$ and $(-2,-7)$ are also identified by
hexagons. For as specified by Theorem \ref{thm:fourth}, coherency also then
requires lines through each of {\it these} points and $(0,0)$ as
well. You will notice that the slope of the line through $(-2,-7)$
does {\it not exceed} the slope of the bounding line
$(1-p_L(8))/p_L(8)\;$; $\ $ whereas the slope of the line through
$(-4,-5)$ {\it does exceed} the slope of the upper bounding line
$(1-p_U(8))/p_U(8)$.  Thus the coherent extension of
$Pa8[2,5,p_L(8),p_U(8)]$ by the assertion $p_{5,9} = 5/9$ is allowed
according to the theorem.  If these limiting slopes had been
exceeded then this coherent extension could only be extended still 
further to bounds $p_L(9)$ and $p_U(9)$ if they were specified to be 
sharper than $p_L(8)$ and $p_U(8)$.  Thus, it must be true that $p_L(9) < 2/9$ and
$p_U(9) > 6/9$. The lines now mentioned completely exhaust the
coherent implications of the presumptions of Theorem \ref{thm:fourth} relevant to
this example.  There are no further implications for a line either
through $(-7, -2)$ nor $(-1,-8)$. Notice that now frequency
mimicking is required for $N=9$ through the frequency domain $[2/9,
6/9]$ which expands the interval $[2/8, 5/8]$ required at $N=8$.  $\
\ \ \ \ \ \ \ \ \ \ \ \ \ \ \ \ \ \ \qed $
\section*{Appendix 4. Supplementary materials}

{\bf Contents:}  The materials in this appendix extend the descriptive discussion of issues that appear in Sections \ref{sect:5.4limit} and \ref{sect:5.5limitfmd} of the article text, including formal statements of Theorems and Proofs relevant to the discussion.
They focus on the construction of distributions inhering agglutinated masses,
and on the limiting distributions for the family of FMD's over a finite interval.
\subsection*{The limit of extended FMD's for a proportion: \\ \hspace*{1cm}distributions exhibiting agglutinated masses}
Exchangeable distributions are most widely known
on account of de Finetti's representation theorem.  It says that if a
sequence of events $E_1, ..., E_{N+1}$ is regarded exchangeably and
as infinitely exchangeably extendible, then for any $a$ and $N$, \vspace{.13cm} \\
%\begin{center}
\hspace*{.75cm} $P(S_{N+1} = a) \ = \ P(\overline{S}_{N+1} = \frac{a}{N+1}) \ = 
\binom{N+1}{a}
%\left({N+1\atop a}\right)
\; \int_0^1 \theta^a\; (1-\theta)^{N+1-a}\;
dM(\theta)$\ , \vspace{.13cm}\\
%\end{center}
where \ $M(\theta) \; = \; \lim_{K\rightarrow\infty} M_{\overline{S}_{N+K}}(\theta)  $\
% \hspace{3.5cm}
for some sequence of finitely additive distributions
\{$M_{\overline{S}_{N+K}}(^.)$\}$_{K=1}^\infty$.\ \;
See Heath and Sudderth (1976), Landenna and Marasini (1986, pp. 87-89) or Lad (1996, pp. 207-209). Diaconis
and Freedman (1980) noted that even if the distribution is
exchangeably extendible only to N+K, then for some such distribution
$M(^.)$ this mixture representation differs from the actual value of
$P(S_{N+1} = a)$ by at most $4N/(N+K)$ for any ``$a$''.  In
practice, the mixing distribution $M(^.)$ in the representation
theorem (or the ``prior distribution for $\theta$'' as it is commonly referred to) is
meant to represent one's initial opinions about the relative
frequency of success in an arbitrarily large sequence of events
one would regard exchangeably with  {\bf E}$_{N+1}$.\\

For any finite subsequence of course, the distribution $M_{N+K}(^.)$ is finitely
additive;  and successive distributions $M_K({\bar{S}_{N+K}})$ in
the sequence must be related by the reduction equations (\ref{eq:usualreduction}) which we
reviewed in Section \ref{sect:reductive}. However, no restrictions are placed on
probability assertions regarding the average of {\it countably} infinite
sequences. The coherency condition specifies only that all
distributions in the sequence are finitely additive.  Thus, the
limit of the finitely additive distributions is not necessarily equal to the
distribution of the limit of ${\bar{S}_{N+K}}$.  The assertion of finite additivity does not constitute a {\it restriction} on the mixing functions considered, but rather a {\it liberation} relative to the axiom of countable additivity which is commonly presumed.  Countably additive distributions are {\it permissible} in the limit, since
they are finitely additive as well. However, de Finetti's insistence
on mere finite additivity as the only meaningful operational
characterisation  of probability allows distributions that extend
the realm of theoretical discussion to wider possibilities,
including those relevant to infinitely extendible FMDs which we
address here.\\

The important distributions in all real problems of
practice are the finite members of the sequence
\{$M_K({\bar{S}_{N+K}})$\}, {\it not} the limit of this sequence.  Nonetheless, in the
context of extendible frequency mimicking distributions
over a limited domain, we can state precisely what happens to the
{\it limiting distribution of the frequencies} $\bar{S}_{N+K}$ as K
increases. (Note again, this is something different from the distribution of the limit of the frequencies.)  Rather than directly stating this result as a
Theorem to be proved, it will be more informative to develop our
understanding through a constructive discussion, with proofs of
algebraic claims within the discussion deferred to the conclusion of the discussion. We begin with a brief informal  overview,
and then state the formalities and proof of a Theorem which develops from a discussion of details.\\

Further to the conditions of Theorem 4 and Corollary 1, which imply
the assertions of \ $Pa(N+K)[a_1, K+a_2, ^., ^.]$ for every $K
\geq 1$, \; any probability mass function denoted by {\bf q}$_{N+K+2}$
is restricted to have only $a_1 + (N - a_2)$ free components. As the
value of $K$ increases, the tendency in all agreeing distributions
is for virtually all the mass in the vector {\bf q}$_{N+K+2}$ to
settle essentially on two points, $S_{N+K+1} = a_1-1$ and $S_{N+K+1}
= a_2+K+1$.  While the probability mass functions for each
$\bar{S}_{N+K}$ are ordinary, the limit of this sequence of
distributions for $\bar{S}_{N+K}$ is an unusual one.  It is improper
and finitely additive, assigning probability $0$ to the points $0$
and $1$ {\it and} to every open interval {\it strictly} within $(0,1)$. However it
exhibits what  \cite{deFinetti1949,deFinetti55}
% de Finetti (1949, 1955)
called ``adherent masses'' at $0$ and $1$
via the assertions of $P(0,a)+P(b,1) = 1$ for any $0 < a < b < 1$.
Recall our Definition \ref{def:first} near the end of Section \ref{sect:intro} of this article. The source and meaning of
such masses emerges from attention to some algebra.\\

To begin the analysis, notice that the positions of the two points
of amassment for the sum $S_{N+K+1}$, that is, $a_1-1$ and
$a_2+K+1$, when divided by $(N+K+1)$, converge to $0$ and to $1$ as
$K$ increases. This is the source of the adherent masses at $0$
and $1$ in the limiting distribution for the sequence of finitely
additive distributions. The total mass of the distribution is
settling on points that are always separated discretely from 0 and 1.
However, the points themselves are being pressed toward these
endpoints of the unit interval, with
gaps between them and the endpoints diminishing.\\

Here is what happens.  ({\it Algebraic details of the following
	statements are presented as a conclusion to this presentation, after
	the resulting Theorem is formalised.}) Each of
the $[K+(a_2-a_1)+1]$ frequency mimicking assertions constituting
those denoted by $Pa(N+K)(a_1, K+a_2, p_L(N+K), p_U(N+K))$ places
one linear restriction on the components of the mass function {\bf
	q}$_{N+K+2}$.   Together with the usual summation restriction, this
means that {\bf q}$_{N+K+2}$ has only $N - (a_2 - a_1)$ free
components.  Suppose we partition all the components of ${\bf q}_{N+K+2}$ into three groups:  the
initial $a_1$ probabilities as $a$ runs from $0$ to $a_1-1$;  the
intermediate $K+(a_2-a_1+1)$
probabilities; and the final $N-a_2+1$ probabilities, thinking of the
first and third groups as constituting the free variables. When $K$
becomes large enough, its free components in the first group are
restricted to be essentially geometrically increasing, while the
third group is restricted to be geometrically decreasing, at
rates proportional to $K$ and $1/K$ respectively. While the first and
third groups of vector components are otherwise relatively free (though
necessarily increasing and decreasing respectively), the sum
of all the interior $q_{a,N+K+1}$ components must equal \vspace{.1cm}\\
\hspace*{.4cm}$q_{a_1}\; \frac{a_1\;(N+K+1-a_1)}{N+K+1}\
[H(a_2+K) - H(a_1-1) + H(N+K-a_1+1) - H(N-a_2)\;]$. \vspace{.1cm}\\
%$q_{a_1}\; \{\;1 + \frac{a_1\; %(N+K+1-a_1)}{N+K+1}\ [H(a_2+K+1) \ -
%\ H(a_1) \ + \ H(N+K-a_1) \ - \ %H(N-a_2-1)]\;\}$. \\
(Recall from Theorem 1 that the function $H(^.)$ is a harmonic sum.)
Since the value of this bracketed $[\;^.\;]$ coefficient on $q_{a_1}$
is unbounded as K increases, the value of $q_{a_1}$ and all
subsequent $q_a$ through $q_{a_2 + K}$ must each deteriorate to $0$.
{\it Moreover, it can be shown that the entire sum of this
	interior series (Group 2) of probability components diminishes toward $0$ as
	well.} Meanwhile, the values of $q_0$ through $q_{a_1-1}$ increase
geometrically by factors of order K. Similarly, the values of
$q_{a_2+K+1}$ through $q_{N+K+1}$ decrease geometrically by factors
of order $1/K$.  Thus, for increasingly large values of K, the probability mass
function comes to be supported essentially only on the points $(a_1
- 1)$ and $(a_2+K+1)$.  This implies that the limiting distribution
of the average, $\bar{S}_{N+K+1}$, becomes uniform at 0 everywhere
on $[0,1]$.  However, it exhibits adherent masses at $0$ and $1$, \ because the
ratios that characterize these points of amassment converge there: \ $(a_1-1)/(N+K+1) \;
\rightarrow \; 0$ \ and \ $(a_2+K+2)/(N+K+1) \;
\rightarrow \; 1$ \ as $K\; \rightarrow \; \infty$.\\

Having introduced this analysis as a discussion, we shall conclude
it with a formal statement of the Theorem that it has motivated, followed
by a complete algebraic proof.\\

\begin{theorem}\label{thm:fifth}
 Further to the conditions of Theorem \ref{thm:fourth} and
Corollary \ref{cor:first}, which imply the assertions of \vspace{.13cm} \\
$Pa(N+K)[a_1, a_2+K, p_L(N+K) \leq a_1/(N+K+1),\vspace{.13cm} p_U(N+K) \geq (a_2+K)/(N+K+1)]$ \\ \ for every $K
\geq 1$, any probability mass function denoted by {\bf q}$_{N+K+2}$
is restricted to have only $N - (a_2 - a_1)$ free components. Even
these are restricted in two groups to exhibit a geometrically
increasing and a geometrically decreasing structure.  As the value
of $K$ increases, all the mass settles essentially on two points,
$S_{N+K+1} = a_1-1$ and $S_{N+K+1} = a_2+K+1$.  While the
probability mass functions for each $\bar{S}_{N+K}$ are ordinary,
the limit of the distributions for $\bar{S}_{N+K}$ is improper and
finitely additive, assigning probability $0$ to the points $0$ and
$1$ and to every open interval within $(0,1)$, but with adherent
masses at $0$ and~$1$.  
%\ \ \ \ \ \ \ \ \ $\diamond$
\end{theorem}
\subsection*{Algebraic proof of Theorem \ref{thm:fifth}:}

The proof revolves upon an algebraic representation of the coherent implications of
asserting $Pa(N+K)[a_1,a_2+K,^.,^.]$ for the pmf {\bf q}$_{N+K+2}$.  Via this representation we can study the limiting distribution for the average number of successes.
For the sake of simplicity in the representation, we shall suppress the second subscript on terms of the form
$q_{a,N+K+1}$ and $p_{a,N+K}$, which merely identify the number of
events under consideration. Here it is always the vector {\bf
	E}$_{N+K+1}$. We print these second subscripts only in the very first
statement of $q_{0,N+K+1}$ as a function of {\bf p}$_{(N+K+1)}$. The
expression for this term incorporates the
summation constraint that all terms $q_{a,N+K+1}$ sum to $1$.\\

To begin, we express the recursive equations $(3)$ 
from Section 2.1 of the main article, 
\begin{equation}
q_{a+1,N+1} \ = \ (\frac{N+1-a}{a+1}) \ (\frac{p_{a,N}}{1 - p_{a,N}}) \ q_{a,N+1} \ , \ \ \ \ {\rm for}\  a = 0, .. . ,N,
\label{eq:recursiveform}
\end{equation}
but applied now to $N+K+1$ events rather than
merely to $N+1$, in three groups: the group sizes are $a_1$, $a_2 + K -
(a_1 -1)$, and $N-a_2+1$. These numbers sum to $N+K+2$, the size of the
pmf vector {\bf q}$_{N+K+2}$.\\

\noindent {\bf Group 1.}  The first term incorporates the summation
constraint, while the following $a_1-1$ relatively unconstrained
terms can be seen to grow geometrically by factors of order K, on account of
the factor $\frac{N+K+a-1}{a}$ which appears in each recursive equation.
Remember that when the integer $i$ is not within $[a_1, a_2]$ the
values of $p_i$ terms are
constrained only to be nondecreasing.\\

$q_{0,N+K+1} \ = \ \{1 + \displaystyle\sum_{a
	=0}^{N+K}
\binom{N+K+1}{a+1}
%\left({N+K+1\atop a+1}\right)
\prod_{i=0}^{a}\
\frac{p_{i,N+K}}{1-p_{i,N+K}}\}^{-1} $\ \ ,\\

$q_1 \ = \ (N+K+1)\ \frac{p_0}{1-p_0} \ q_0$\\

$q_2 \ = \ \frac{N+K}{2}\ \frac{p_1}{1-p_1} \ q_1$\\

$q_3 \ = \ \frac{N+K-1}{3}\ \frac{p_2}{1-p_2} \ q_2$\\

$^.$.$_.$\\

$q_{a_1-1} \ = \ \frac{N+K+1-(a_1-2)}{a_1-1}\ \frac{p_{a_1-2}}
{1-p_{a_1-2}} \ q_{a_1-2}$ \ \ \ \ \ \ \ .\\

\noindent {\bf Group 2.}  The first equality for each $q$ in this
group of $a_2 + K - a_1 + 1$ terms continues this format of
recursive representations.  The second equality
replaces each of the relevant
$p_{a,N+K}$ assertions with their frequency mimicking  values, $a/(N+K)$,
in the recursive form, and simplifies the algebraic expression so that patterns can be seen:\\

\hspace{.25cm}$q_{a_1} \ \ = \ \frac{N+K+1-(a_1-1)}{a_1}\ \frac{p_{a_1-1}}
{1-p_{a_1-1}} \ q_{a_1-1}$\\

$q_{a_1+1} \ = \ \frac{N+K+1-a_1}{a_1+1}\ \frac{p_{a_1}}{1-p_{a_1}}
\ q_{a_1} \ \ = \ \ \frac{N+K-a_1+1}{a_1+1}\ \frac{a_1}{N+K-a_1} \
q_{a_1}$\\

$q_{a_1+2} \ = \ \frac{N+K+1-(a_1+1)}{a_1+2}\
\frac{p_{a_1+1}}{1-p_{a_1+1}} \ q_{a_1+1} \ \ = \ \
\frac{N+K-a_1}{a_1+2}\ \frac{a_1+1}{N+K-(a_1+1)}
\ q_{a_1+1}$\\

$^.$.$_.$\\

$q_{a_2+K-1} \ = \ \frac{N+K-(a_2+K-2)}{a_2+K-1}\
\frac{p_{a_2+K-2}}{1-p_{a_2+K-2}} \ q_{a_2+K-2} \ \ = \ \
\frac{N-a_2+3}{a_2+K-1}\ \frac{a_2+K-2}{N-a_2+2} \ q_{a_2+K-2}$\\

$q_{a_2+K} \ = \ \frac{N+K+1-(a_2+K-1)}{a_2+K}\
\frac{p_{a_2+K-1}}{1-p_{a_2+K-1}} \ q_{a_2+K-1} \ \ = \ \
\frac{N-a_2+2}{a_2+K}\ \frac{a_2+K-1}{N-a_2+1} \ q_{a_2+K-1}$ \ \ .\\

\noindent {\bf Group 3.}  Continuing these recursive
expressions, although the third group of $N-a_2+1$ terms are again
relatively unrestricted, they eventually
decrease geometrically by factors on the order of $1/K$.  This is evident
on account of the factor $\frac{N-a_2+1}{a_2+K+1}$ which appears in each
simplified recursive equation.
(The three lines showing a second equality are merely algebraic  simplifications.)\\

$q_{a_2+K+1} \ = \ \frac{N+K+1-(a_2+K)}{a_2+K+1}\
\frac{p_{a_2+K}}{1-p_{a_2+K}} \ q_{a_2+K} \ \ = \ \
\frac{N-a_2+1}{a_2+K+1}\ \frac{p_{a_2+K}}{1-p_{a_2+K}}
\ q_{a_2+K}$\\

$q_{a_2+K+2} \ = \ \frac{N+K+1-(a_2+K+1)}{a_2+K+2}\
\frac{p_{a_2+K+1}}{1-p_{a_2+K+1}} \ q_{a_2+K+1}
\ \ = \ \
\frac{N-a_2}{a_2+K+2}\ \frac{p_{a_2+K+1}}{1-p_{a_2+K+1}}
\ q_{a_2+K+1}$\\

$q_{a_2+K+3} \ = \ \frac{N+K+1-(a_2+K+2)}{a_2+K+3}\
\frac{p_{a_2+K+2}}{1-p_{a_2+K+2}} \ q_{a_2+K+2} \ \ = \ \
\frac{N-a_2-1}{a_2+K+3}\ \frac{p_{a_2+K+2}}{1-p_{a_2+K+2}}
\ q_{a_2+K+2}$\\

$^.$.$_.$ \ \ until \\

$q_{N+K} \ = \ \frac{2}{N+K}\ \frac{p_{N+K-1}}{1-p_{N+K-1}}
\ q_{N+K-1}$ \ ; \ \ and finally,\\

$q_{N+K+1} \ = \ \frac{1}{N+K+1}\ \frac{p_{N+K}}{1-p_{N+K}} \ q_{N+K}$ \ .\\

{\bf We begin the analysis now by examining  the second group of $q_a$ values,} in particular
the {\it sum} of these $q_a$'s.  Having substituted the frequency
mimicking values of $p_i$ with $i/(N+K)$, and now substituting the
recursive multiplicands $q_{a_1+i}$ with their expressions in terms of
$q_{a_1}$ using equation (3), the sum of all the constrained
terms in the second group becomes \\

%$q_{a_1}\ [\ 1 \ + \ a_1\ (N+K+1-a_1)\ \{ %\frac{1}{(a_1+1)(N+K-a_1)}
%\ + \
%\frac{1}{(a_1+2)(N+K-a_1-1)} \ + \ ^{...} %\ + \\
%\hspace*{7.5cm}+ \ \frac{1}{(a_2+K)(N+1
%-a_2)}\ + \ \frac{1}{(a_2+K+1)(N-a_2)}\} %\ ]$ \\
%
%$= \ q_{a_1}\ [\ 1 \ + \ a_1\ (N+K+1-a_1
%)\ \Sigma_{i = a_1+1}^{a_2+K+1} \ %\frac{1}{i\ (N+K+1-i)} \ ]$ \\
%
%$= \ q_{a_1}\ [\ 1 \ + \ \frac{a_1\ (N+K+1-a_1)}{N+K+1}\ \{\Sigma_{i
%= a_1+1}^{a_2+K+1}\ \frac{1}{i} \ + \
%\Sigma_{i = a_1+1}^{a_2+K+1} \frac{1}{N+K+1-i}\} \ ] $\\
%
%$= \ q_{a_1}\ [\ 1 \ + \ \frac{a_1\ (N+K+1-a_1)}{N+K+1}\
%\{H(a_2+K+1) - H(a_1) + H(N+K-a_1) - H(N-a_2-1) \}]$ . (A)\\

%WORKING ON THIS ADJUSTMENT \\

$q_{a_1} \ a_1\ (N+K+1-a_1)\ [ \frac{1}{a_1(N+K+1-a_1)}
\ + \ \frac{1}{(a_1+1)(N+K-a_1)}
\ + \
\frac{1}{(a_1+2)(N+K-a_1-1)} \ + \ ^{...} \ + \\
\hspace*{7.5cm}+ \ \frac{1}{(a_2+K-1)(N-a_2+2)}\ + \ \frac{1}{(a_2+K)(N-a_2+1)}] \ $ \\

$= \ q_{a_1}\ a_1\ (N+K+1-a_1)\ \Sigma_{i = a_1}^{a_2+K} \ \frac{1}{i\ (N+K+1-i)} \ $ \\

$= \ q_{a_1}\ \frac{a_1\ (N+K+1-a_1)}{N+K+1}\ [\;\Sigma_{i
	= a_1}^{a_2+K}\ \frac{1}{i} \ + \
\Sigma_{i = a_1}^{a_2+K} \frac{1}{N+K+1-i}\;]  $\\

$= \; q_{a_1}\; \frac{a_1\;(N+K+1-a_1)}{N+K+1}\
[H(a_2+K) - H(a_1-1) + H(N+K-a_1+1) - H(N-a_2)\;]$ . (A)\\

\noindent Because of the similarity of the harmonic sum $H(^.)$ to the natural logarithm as the value of $K$ increases, the expression in square brackets seen in line (A) becomes not too different from $log(a_2+K) \;+\; log(N+K-a_1+1)$ which is unbounded.  This implies that $q_{a_1} \rightarrow 0$ as $K
\rightarrow \infty$.  Moreover, the
value of all subsequent $q_a$ terms in the second itemised group
converge to $0$ as well, on account of their recursive relation to $q_{a_1}$.\\

Furthermore, the {\it sum} of all terms in the second
group converges to $0$ too as $K$ increases.  This follows from an argument similar
to our proof of Theorem 2.  Firstly, the iterative use of the
reduction equation $(4)$ in Section 2.2 of the article will reduce the mass
vector {\bf q}$_{N+K+2}$ to a vector of any lower fixed dimension, {\bf
	q}$_{N+M+2}$. Then summing the components of this {\bf q}$_{N+M+2}$
over its component members $a_1, a_1+1, ..., a_2+N+M+1$ (the region
of frequency mimicking) would yield a bounded multiple of $q_{a_1,
	N+M+1}$ similar to equation (7) in the proof of Theorem 2.  Thus, for any
positive value of $M$, this sum of terms in the second group converges to 0 as $K$ increases.  In fact, Theorem 2 can now be seen as a special case of this result, supposing $a_1 = 1$ and $a_2 = N+K-1$.\\

Now in contrast to the second intermediate group of ordered $q$'s
whose sum is tending toward $0$, notice in  our first group sequence of
formulae that the values of $q_0, q_1, ..., q_{a_1-1}$ must be
ascending as a geometric progression, each of them augmenting in
size by a factor on the order of $N+K$. Similarly, the values in the
third group of $q_{a_2+K+2}$ through $q_{N+K+1}$ must eventually be
descending geometrically by a factor of a similar order, $1/(N+K)$.  Thus, as
$K \rightarrow \infty$, the entire distribution of $S_{N+K+1}$
becomes essentially amassed on the end-point terms of these groups,
term $a_1-1$ and term $a_2+K+1$. The first of these, term
$a_1-1$ is the $a_1^{\rm th}$ smallest possibility for the sum, while
the latter term $a_2+K+1$  is $(N-a_2)^{th}$ largest possible value of the sum.
When the value of the sum $S_{N+K+1}$ is divided by ${N+K+1}$,
the masses at these two positions for the proportion
$\bar{S}_{N+K+1}$ become forced toward the boundaries of the interval
$(0,1)$.  These sequence limits which are separated from the interval endpoints, $0$ and $1$, are
the source of adherent masses at $0$
and $1$ in the finitely additive limiting distribution for
$\bar{S}_N$. \ \ \ \ \ \ \ \ \ \ \
\ \ $\qed$

\subsection*{The limiting boundary of FMD's over a constrained interval}

Having reached this conclusion about the applicability of FMD's to finite population problems, it is intriguing to investigate the limiting distribution for the family of finite distributions agreeing with the assertions $PaN[a_1,a_2,p_L(N),p_U(N)]$ for any $N$.  Consider the limit of distributions that respect frequency mimicking behaviour over the largest rational interval within a constant real interval $[\theta_1,\theta_2]$ as the size of N increases.  An analysis that parallels the algebraic representations in the proof of Theorem 5, with an important modification, yields the result that in the case of ``Strict'' or ``Linear'' extensions outside the FMD interval this limiting distribution is mixture Binomial with respect to a 4-parameter {\it Incomplete Beta} mixing function with a constrained domain.\\

In this case it will be simplest to state the theorem and then to describe the structure of its proof before getting into precise details.  These will then rely on the algebraic representations derived in the proof of Theorem 5, but apply them with a slight twist.

\begin{theorem} For any choice of fixed real values for $\theta_1 < \theta_2$, each within $(0,1)$, define $a_1(\theta_1,N) \equiv [[N\theta_1$ as the smallest integer $a$ for which the rational number $a/N \geq \theta_1$, and $a_2(\theta_2,N) \equiv \ N\theta_2]]$ as the largest integer $a$ for which $a/N \leq \theta_2.$  The limiting distribution for the family of finite FMD's specified by $PaN[a_1(\theta_1,N), a_2(\theta_2,N),^.,^.]$ along with ``Strict'' or ``Linear'' extensions of conditional probabilities outside of $(\theta_1, \theta_2)$ is a 4-parameter Incomplete Beta mixture of Binomial distributions, with Incomplete Beta mixing parameters $(\theta_1, \theta_2, 0, 0)$. 
%	\ \ \ \ \ \ \ \ \ $\diamond$\\
\end{theorem}
\noindent{\bf Comments:}  Notice firstly that if the specification of either $\theta_i$ is an irrational number, then the definition of the associated $a_i(\theta_i, N)$ is always a well defined integer for which $a_i/N$ appropriately exceeds or falls short of this $\theta_i$.  On the other hand, if the value of $\theta_i$ is rational, then the associated value of the rational number $a_i(\theta_i, N)/N$ will differ from $\theta_i$ only when their minimum denominators are incommensurable.  If the value of $\theta_1$ equals $J/K$ for some integers $J$ and $K$ without common factors, for example, the value of $a_1(\theta_1,N)/N$ will equal $\theta_1$ if and only if N is an integer multiple of $K$.  Otherwise it will be the smallest value of $J/K$ that exceeds $\theta$.\\

Secondly, recall that while a $Complete Beta(\alpha, \beta)$ density function for $\theta$, which is proportional to $\theta^{\alpha-1}(1-\theta)^{\beta-1}$ over $\theta \in [0,1]$, allows only parameters $\alpha > 0$ and $\beta > 0$ for proper integration, the four parameter Incomplete density over $\theta \in (\theta_1,\theta_2)$ strictly within $(0,1)$ integrates naturally when $\alpha = 0$ and $\beta = 0$.  In such a case the proportionality constant for the density equals $\{log[\theta_2/(1-\theta_2)] - log [\theta_1/(1-\theta_1)]\}^{-1}$.  Although this density is zero outside the interval $(\theta_1,\theta_2)$, when it mixes corresponding Binomial distributions as prescribed by exchangeability, the mixture allows positive probabilities for appropriate rational values of the average successes across the entire spectrum of rationals within $[0,1]$.\\

Thirdly, we mention that we have only presented here the limiting result for the Strict or Linear Extensions of $PaN[a_1(\theta_1,N), a_2(\theta_2,N),^.,^.]$ assertions.  We have achieved results for the Quadratic and Weak Extensions as well.  However, their details are more complicated, defying fruitful presentation here.\\

\noindent{\bf Proof Structure:}   Algebraic details of the following claims appear below.  As in the proof of Theorem 5, the probabilities for the possible values of the sum $S_{N+1}$ are partitioned into three groups.  In the case of Strict extensions, the first group turn out to be probability masses for the average variable value of a $Binomial(N,\theta_1)$ distribution.  However, they are summed only over values of the variable that are strictly less than $\theta_1$.  Since the average of a $Binomial(N,\theta_1)$ variable converges almost surely to $\theta_1$ itself, this sum of probabilities over smaller values surely converges to $0$. In the case of Linear extensions the sum of associated probabilities less than $\theta_1$ turns out to be even smaller, with the same consequences.   (A similar argument applies to Group 3 probabilities.)  In contrast, when frequency mimicking conditional probabilities are inserted into the recursive equations for Group 2 probabilities, each 
component probability relative to its bin width is found to converge to the density value of an $Incomplete Beta (\theta_1, \theta_2, 0, 0)$ distribution. \ \ \ \ \ \ \ $\qed $

\subsection*{Algebraic details of the proof of Theorem 6}

To begin, we shall again express the recursive equations
(\ref{eq:recursiveform})  in three groups, but applied now to $N+1$ events:
the group sizes are $a_1$, $a_2 -
a_1 +1$, and $N-a_2+1$. These three numbers sum to $N+2$, the size of the
pmf vector {\bf q}$_{N+2}$.  However we need remember that now  these values of $a_1$ and $a_2$ depend on $N$ and $(\theta_1, \theta_2)$, viz., $a_1 = a_1(\theta_1,N) \ \equiv \ [[N\theta_1$, and $a_2 = a_2(\theta_2,N) \ \equiv \ N\theta_2]]$.  Moreover, $a_1/N \rightarrow \theta_1$ and $a_2/N \rightarrow \theta_2$ as $N \rightarrow \infty$. \\

\noindent {\bf Group 1.}  The first term incorporates the summation
constraint, while the following $a_1-1$ relatively unconstrained
terms can be seen to grow geometrically by factors of order N.
Remember that when the integer $i$ is not within $[a_1, a_2]$ the
values of $p_i$ terms are
constrained only to be nondecreasing. \\

$q_{0,N+1} \ = \ \{1 + \displaystyle\sum_{a
	=0}^{N} 
\binom{N+1}{a+1}
%\left({N+1\atop a+1}\right) 
\prod_{i=0}^{a}\
\frac{p_{i,N}}{1-p_{i,N}}\}^{-1} $\ \ , \\

$q_1 \ = \ (N+1)\ \frac{p_0}{1-p_0} \ q_0$\\

$q_2 \ = \ \frac{N}{2}\ \frac{p_1}{1-p_1} \ q_1$\\

$q_3 \ = \ \frac{N-1}{3}\ \frac{p_2}{1-p_2} \ q_2$\\

$^.$.$_.$\\

$q_{a_1-1} \ = \ \frac{N+1-(a_1-2)}{a_1-1}\ \frac{p_{a_1-2}}
{1-p_{a_1-2}} \ q_{a_1-2}$ \ \ \ \ \ \ \ .\\

\noindent {\bf Group 2.}  The first equality for each $q$ in this
group of $a_2 - a_1 + 1$ terms continues this format of
recursive representations.  The second equality
replaces each of the relevant
$p_{a,N}$ assertions with their frequency mimicking  values, $a/N$, and simplifies the algebraic expression so that patterns can be seen:\\

$q_{a_1} \ = \ \frac{N+1-(a_1)}{a_1}\ \frac{p_{a_1-1}}
{1-p_{a_1-1}} \ q_{a_1-1}$\\

$q_{a_1+1} \ = \ \frac{N+1-a_1}{a_1+1}\ \frac{p_{a_1}}{1-p_{a_1}}
\ q_{a_1} \ \ = \ \ \frac{N-a_1+1}{a_1+1}\ \frac{a_1}{N-a_1} \
q_{a_1}$\\

$q_{a_1+2} \ = \ \frac{N+1-(a_1+1)}{a_1+2}\
\frac{p_{a_1+1}}{1-p_{a_1+1}} \ q_{a_1+1} \ \ = \ \
\frac{N-a_1}{a_1+2}\ \frac{a_1+1}{N-(a_1+1)}
\ q_{a_1+1}$\\

$^.$.$_.$\\

$q_{a_2-1} \ = \ \frac{N+1-(a_2-2)}{a_2-1}\
\frac{p_{a_2-2}}{1-p_{a_2-2}} \ q_{a_2-2} \ \ = \ \
\frac{N-a_2+3}{a_2-1}\ \frac{a_2-2}{N-a_2+2}
\ q_{a_2-2}$\\

$q_{a_2} \ = \ \frac{N+1-(a_2-1)}{a_2}\
\frac{p_{a_2-1}}{1-p_{a_2-1}} \ q_{a_2-1} \ \ = \ \
\frac{N-a_2+2}{a_2}\ \frac{a_2-1}{N-a_2+1} \ q_{a_2-1}$ \ \ .\\

\noindent {\bf Group 3.}  Continuing the list of recursive
expressions, though the third group of $N-a_2+1$ terms are again
relatively unrestricted, they eventually
decrease geometrically by factors on the order of N.
(The three lines with a second equality are merely simplifications.)\\

$q_{a_2+1} \ = \ \frac{N+1-(a_2)}{a_2+1}\
\frac{p_{a_2}}{1-p_{a_2}} \ q_{a_2} \ \ = \ \
\frac{N-a_2+1}{a_2+1}\ \frac{p_{a_2}}{1-p_{a_2}}
\ q_{a_2}$\\

$q_{a_2+2} \ = \ \frac{N+1-(a_2+1)}{a_2+2}\
\frac{p_{a_2+1}}{1-p_{a_2+1}} \ q_{a_2+1}
\ \ = \ \
\frac{N-a_2}{a_2+2}\ \frac{p_{a_2+1}}{1-p_{a_2+1}}
\ q_{a_2+1}$\\

$q_{a_2+3} \ = \ \frac{N+1-(a_2+2)}{a_2+3}\
\frac{p_{a_2+2}}{1-p_{a_2+2}} \ q_{a_2+2} \ \ = \ \
\frac{N-a_2-1}{a_2+3}\ \frac{p_{a_2+2}}{1-p_{a_2+2}}
\ q_{a_2+2}$\\

$^.$.$_.$ \ \ until\\

$q_{N} \ = \ \frac{2}{N}\ \frac{p_{N-1}}{1-p_{N-1}}
\ q_{N-1}$ \ ;\\

$q_{N+1} \ = \ \frac{1}{N+1}\ \frac{p_{N}}{1-p_{N}} \ q_{N}$ \ .\\

{\bf We shall now examine the second group of $q_a$ values,} in particular
the {\it sum} of these $q_a$'s.  Having substituted the frequency
mimicking values of $p_i$ with $i/N$, and now substituting the
recursive multiplicands $q_{a_1+i}$ with their expressions in terms of
$q_{a_1}$ using equation (3), the sum of all the constrained
terms in the second group becomes\\

$q_{a_1}\ [\ 1 \ + \ a_1\ (N+1-a_1)\ \{ \frac{1}{(a_1+1)(N-a_1)}
\ + \
\frac{1}{(a_1+2)(N-a_1-1)} \ + \ ^{...} \ + \\
\hspace*{7.5cm}+ \ \frac{1}{(a_2-1)(N+2-a_2)} \ + \ \frac{1}{(a_2)(N+1-a_2)}\; \} \ ]$\\

$= \ q_{a_1}\ [\ a_1\ (N+1-a_1)\ \Sigma_{i = a_1}^{a_2} \ \frac{1}{i\ (N+1-i)} \ ]$\\

$= \ q_{a_1}\  \frac{a_1\ (N+1-a_1)}{N+1}\ [\;\Sigma_{i
	= a_1}^{a_2}\ \frac{1}{i} \ + \
\Sigma_{i = a_1}^{a_2} \frac{1}{N+1-i}\;]$\\

$= \ q_{a_1}\ \frac{a_1\ (N+1-a_1)}{N+1}\
[\;H(a_2) - H(a_1-1) + H(N+1-a_1) - H(N-a_2) \;]$ . (A)\\

\noindent Since the harmonic sum function $H(N)$ is unbounded as
$N$ increases, this coefficient on $q_{a_1}$ is unbounded too.  Because of the similarity of the harmonic sum to the natural logarithm for large $N$, the expression in brackets $[\;]$ becomes not too different from $log(a_2) - log(a_1-1)\;+\; log(N-a_1+1)\;-\;log(N-a_2)$.  Equivalently, for large values of N, this term in brackets ${}$ becomes indistinguishable from\\

$log(\frac{a_2}{a_1-1}) \ + \ log(\frac{N-a_1+1}{N-a_2}) \ = \ log(\frac{a_2/(N+1)}{(a_1-1)/(N+1)}) \ + \ log(\frac{(N-a_1+1)/(N+1)}{(N-a_2)/(N+1)})$ \ \ \ ,\\

and thus it converges to\\

$log(\frac{\theta_2}{\theta_1}) \ + \ log(\frac{1-\theta_1}{1-\theta_2}) \ = \ log[\frac{(1-\theta_1)}{\theta_1}\frac{\theta_2}{(1-\theta_2)}]$ \ \ \ .\\

\noindent Thus, the complete coefficient on $q_{a_1}$ gets close to  $[ (N+1)\theta_1(1-\theta_1)]\;log[\frac{(1-\theta_1)}{\theta_1}\frac{\theta_2}{(1-\theta_2)}]$\;.\\

Well this surely grows with $N$, so $q_{a_1}$ converges to $0$.  However in this case $q_{a_1}$ goes to $0$ in such a way that its product with the bracketed coefficient goes to $1$.  Moreover, the coefficient on $q_{a_1}$ begins to look like the Incomplete Beta function value.  This can be learned by studying the behaviours of the sums of Group 1 and Group 3 probabilities, which we shall do now.

\subsection*{  Limit sum of Group 1 in the ``Strict Case''}

Remember that the ``strict'' completion of the pmf vector {\bf q}$_{N+2}$ derives from the specification of $p_{a,N+1} = a_1(N,\theta)/N$ for each $a = 0, 1, ..., (a_1-1)$.  In this context
the value of each such $p_a/(1-p_a)$ in the Group 1 equations converges to $A_1 \equiv \theta_1/(1-\theta_1)$;  and furthermore, for large values of $N$, the values of the associated $q_{a, N+1}$ become indistinguishable from\\

$q_{0,N+1} \ = \ q_0$\ \ ,\\

$q_1 \ = \ (N+1)\ A_1 \ q_0$\\

$q_2 \ = \ \frac{N}{2}\ A_1 \ q_1 \ = \ ^{N+1}C_{2} \ A_1^{\ 2} \ q_0 $\\

$q_3 \ = \ \frac{N-1}{3}\ A_1 \ q_2 \ = \ ^{N+1}C_{3} \ A_1^{\ 3} \ q_0 $\\

$^.$.$_.$\\

$q_{a_1-1} \ = \ \frac{N+1-(a_1-2)}{a_1-1}\ \frac{p_{a_1-2}}
{1-p_{a_1-2}} \ q_{a_1-2} \ = \ ^{N+1}C_{a_1-1} \ A_1^{\ (a_1-1)} \ q_0 $ \ \ \ \ \ \ \ .\\

\noindent Thus, the sum of terms in Group 1  becomes near to\\

$q_0 \ \displaystyle\sum_{a=0}^{a_1-1} \  ^{N+1}C_a \ A_1^a \ \  = \ \ \frac{q_0}{(1-\theta_1)^{(N+1)}} \ \displaystyle\sum_{a=0}^{a_1-1}  \ \ ^{N+1}C_a \ \theta_1^a \; (1-\theta_1)^{(N+1-a)}\ . $\\

\noindent Since the expression for $q_0$ incorporates the summation constraint on all the $q$'s, this sum of the first ``$a$'' terms depends only on the summation
$\displaystyle\sum_{a=0}^{a_1-1}  \ \ ^{N+1}C_a \ \theta_1^a \; (1-\theta_1)^{(N+1-a)}\ . $\\

Well, this summation represents the sum of probability values for a quantity $S$ that is distributed as $Binomial(N+1, \theta_1)$, viz., $P[S \leq (N+1)\theta_1 - 1] \ \ = \ \ P[\frac{S}{N+1} < \theta_1]$.
To conclude, then, the law of large numbers tells us that this probability converges to $0$, because the proportion $\frac{S}{N+1}$ converges almost surely to $\theta_1$ itself.\\

As to Group 3, the structure of the pmf components can be seen to be  identical to that of Group 1 components, but applied to the negated events $\tilde{E}_a$.  Thus the sum of the Group 3 probabilities also converges to $0$.\\

Since the sum of Group 2 probabilities converges to $1$ while each particular component converges to $0$, the vector of individual $q$'s in Group 2, when divided by $N$ converge to a limiting density function for $\bar{S}_N$  that is identifiable as an $Incomplete\ Beta (\theta_1,\theta_2,0,0)$ density.

\subsection*{  Limit sum of Group 1 in the ``Linear Case''}

The analysis of Group 1 probabilities in the case of  ``Linear'' extensions of conditional probabilities outside the FMD interval $(a_1(\theta_1,N), a_2(\theta_2,N))$ begins just as in the case of ``Strict'' extensions.  However, the summation of probabilities that is assessed there is now no longer a sum of $Binomial(N+1, \theta_1)$ probability masses;  for in this case, each of the expressions for summand probabilities in Group 1 involves a conditional odds ratio $p_i/(1-p_i)$ that is smaller than $\theta_1$.  Remember that the values of $p_i$ in this Group increase linearly from $p_L(N)$ to $a_1(\theta_1,N)/N$.  Thus, in the limit they increase from the limit $p_L(N)$ to $\theta_1$.  Since each of these probabilities in the resulting summation expression is smaller than the Binomial probability masses in the ``Strict'' case, the total summation must be even smaller than in that case.  Thus, it too converges to $0$ as $N$ increases.  As a result, the limiting distribution of the 
$q's$ tends to the same $Incomplete\ Beta$ mixture of Binomials as it does it the ``Strict'' extension.
\ \ $\qed$

\end{document}